\documentclass[11pt]{article}
\usepackage{fullpage}
\usepackage[top=0.9in,bottom=1.1in,left=1.1in,right=1.1in]{geometry}
\usepackage{graphicx}
\usepackage{amsmath,amssymb,graphicx,paralist,amsthm}
\usepackage{algorithm,algorithmic}
\RequirePackage{hyperref}

\usepackage{times}
\usepackage{bm}
\usepackage{natbib}
\usepackage{anyfontsize}
\usepackage{subcaption}
\usepackage{authblk}

\newtheorem{condition}{Condition}
\newtheorem{definition}{Definition}
\newtheorem{lemma}{Lemma}
\newtheorem{theorem}{Theorem}

\newcommand{\fbf}[1]{\mathbf{#1}}
\newcommand{\bE}{\fbf{E}}
\newcommand{\cB}{\mathcal{B}}

\newcommand{\cF}{\mathcal{F}}
\newcommand{\cN}{\mathcal{N}}

\newcommand{\cX}{\mathcal{X}}
\newcommand{\dd}{\mathop{}\!\mathrm{d}}
\newcommand{\EE}{\mathbf{E}}
\newcommand{\PP}{\mathbf{P}}

\newcommand{\rad}{\operatorname{Rad}_n}
\newcommand{\erad}{\widehat{\operatorname{Rad}}}
\def\iid{{\em i.i.d.~}} 
\newcommand{\wrt}{with respect to }
\newcommand{\var}{\mathbf{Var}}
\newcommand{\cov}{\mathbf{Cov}}

\newcommand{\sumin}{\sum_{i=1}^n}

\newcommand{\ignore}[1]{}

\title{A Generalized Fréchet Test for Object Data with Unequal Repeated Measurements}

\author[1]{Jingru Zhang\footnote{jr\_zhang@fudan.edu.cn}}
\author[2]{Shengjie Zhang}
\author[3]{Christopher W Jones}
\author[3]{Mathias Basner}
\author[4]{Haochang Shou}
\affil[1]{School of Data Science, Fudan University}
\affil[2]{Institute of Science and Technology for Brain-inspired Intelligence, Fudan University}
\affil[3]{Department of Psychiatry, University of Pennsylvania}
\affil[4]{Department of Biostatistics, Epidemiology and Informatics, University of Pennsylvania}

\date{}

\begin{document}
\maketitle

\begin{abstract}
Advancements in data collection have led to increasingly common repeated observations with complex structures in biomedical studies. 
Treating these observations as random objects, rather than summarizing features as vectors, avoids feature extraction and better reflects the data's nature. 
Examples include repeatedly measured activity intensity distributions in physical activity analysis and brain networks in neuroimaging. 
Testing whether these repeated random objects differ across groups is fundamentally important; however, traditional statistical tests often face challenges due to the non-Euclidean nature of metric spaces, dependencies from repeated measurements, and the unequal number of repeated measures.
By defining within-subject variability using pairwise distances between repeated measures and extending Fréchet analysis of variance, we develop a generalized Fréchet test for exchangeable repeated random objects, applicable to general metric space-valued data with unequal numbers of repeated measures. 
The proposed test can simultaneously detect differences in location, scale, and within-subject variability.
We derive the asymptotic distribution of the test statistic, which follows a weighted chi-squared distribution. Simulations demonstrate that the proposed test performs well across different types of random objects. We illustrate its effectiveness through applications to physical activity data and resting-state functional magnetic resonance imaging data.
\end{abstract}

Key words:
Imbalanced repeated measurements;
Neuroimaging;
Nonparametric test;
Non-Euclidean data;
Physical activity;
Within-subject variability

\section{Introduction}
Repeated measures are commonly used to capture within-subject variability and improve data reliability.
With the growing complexity of biomedical research and advancements in data collection technologies, it is increasingly common to encounter repeated random object data with unequal numbers of measurements in non-Euclidean metric spaces.
For instance, in the Sleep and Altertness substudy of the Individualized Comparative Effectiveness of Models Optimizing Patient Safety and Resident Education (iCOMPARE) trial, medical interns were randomly assigned to either flexible duty-hour programs or standard duty-hour programs \citep{basner2011validity,basner2019sleep}.
They were instructed to wear an actigraph device that continuously tracked their daily physical activities and complete a brief smartphone survey each morning. 
The survey included a sleep log, a sleep quality score, the Karolinska Sleepiness Scale score, and a brief psychomotor vigilance test. 
In this case, each individual had two sets of outcomes each day: minute-by-minute 24-hour activity profiles and sleep-related measures. While these sleep-related outcomes are usually summarized as a Euclidean vector, the daily activity profiles are multivariate objects that follow complicated data distributions and cannot be appropriately described and modeled as a Euclidean vector.
Instead summarizing the activity data into a few lower dimensional statistics that might be sensitive to the choices of processing algorithms and cutpoints, recent attention has shifted toward using directly modeling the continuous distribution of daily activity intensities \citep{keadle2014impact,schrack2016assessing,yang2020quantlet,zhang2024mediation}, as it more accurately reflects the overall activity level. 
As an illustration, Figure \ref{fig:example-icompare}A shows the minute-by-minute activity counts for one participant across three days in the iCOMPARE trial, and Figure \ref{fig:example-icompare}B presents corresponding histograms of the log-transformed activity counts. 
Figure \ref{fig:example-icompare}C displays bar plots of numbers of repeated measures (days) for participants in the standard and flexible duty-hour groups.
The number of repeated measures varies across subjects, and the distribution of the number of repeated measures is notably different between the two groups, indicating a clear imbalance in measurement frequency.
Similar challenges arise in the neuroimaging field, where assessing test-retest reliability and reproducibility for functional and structural connectomics is crucial \citep{anderson2011reproducibility,zuo2014test,noble2017influences}.
To address this, the Consortium for Reliability and Reproducibility (CoRR) has aggregated previously collected test-retest imaging datasets from more than 36 laboratories around the
world \citep{zuo2014open}. 
These laboratories recruited participants to undergo repeated imaging scans within a few hours.
Particularly, resting-state functional magnetic resonance imaging (rs-fMRI) is widely recognized as an imaging modality with a low signal-to-noise ratio. Hence, it has become increasingly common for research labs to acquire test-retest rs-fMRI data, resulting in two or more brain connectivity networks for each subject.
For instance, Figure \ref{fig:example-neuro}A illustrates brain functional connectivity networks from two scans of a randomly selected individual, with the corresponding graph Laplacians shown in Figure \ref{fig:example-neuro}B. Figure \ref{fig:example-neuro}C presents bar plots depicting numbers of repeated measures (scans) for young and middle-aged participants from a study conducted at New York University Langone Medical Center, as aggregated by CoRR.
Although most subjects underwent two scans, the number of repeated measures differs among individuals, clearly indicating imbalanced repeated measurements.
Analyzing these repeated objects requires robust statistical methods capable of handling non-Euclidean data and accommodating unequal repeated measurements.
\begin{figure}[!h]
	\centering
	\includegraphics[width=\textwidth]{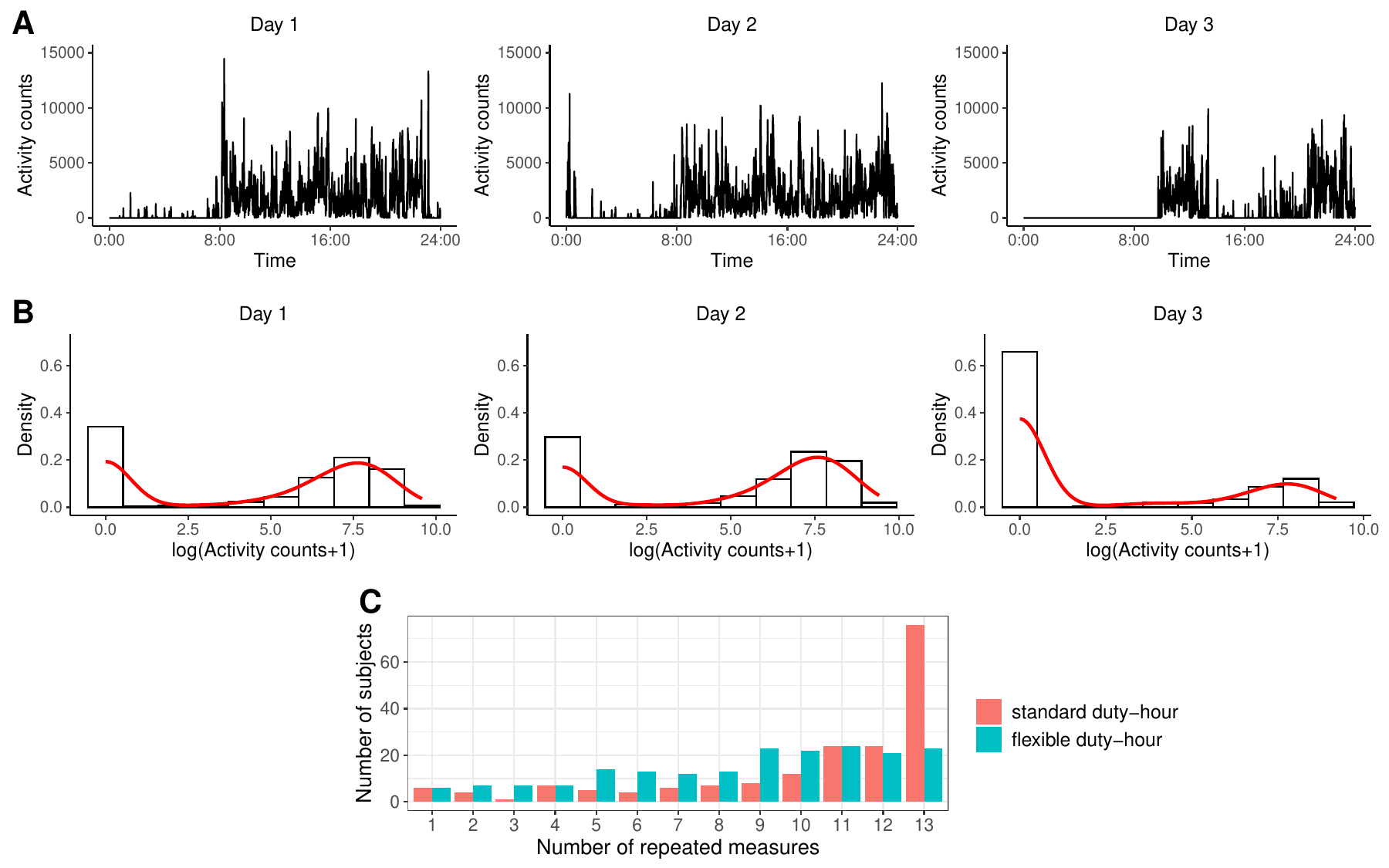}
	\caption{
	(A) Trends of activity counts for a randomly selected individual over three days.
(B) Corresponding histograms and density plots of log-transformed activity counts.
(C) Bar plots illustrating numbers of repeated measures for participants in the standard duty-hour group and the flexible duty-hour group.
	}\label{fig:example-icompare}
\end{figure}
\begin{figure}[!h]
	\centering
	\includegraphics[width=\textwidth]{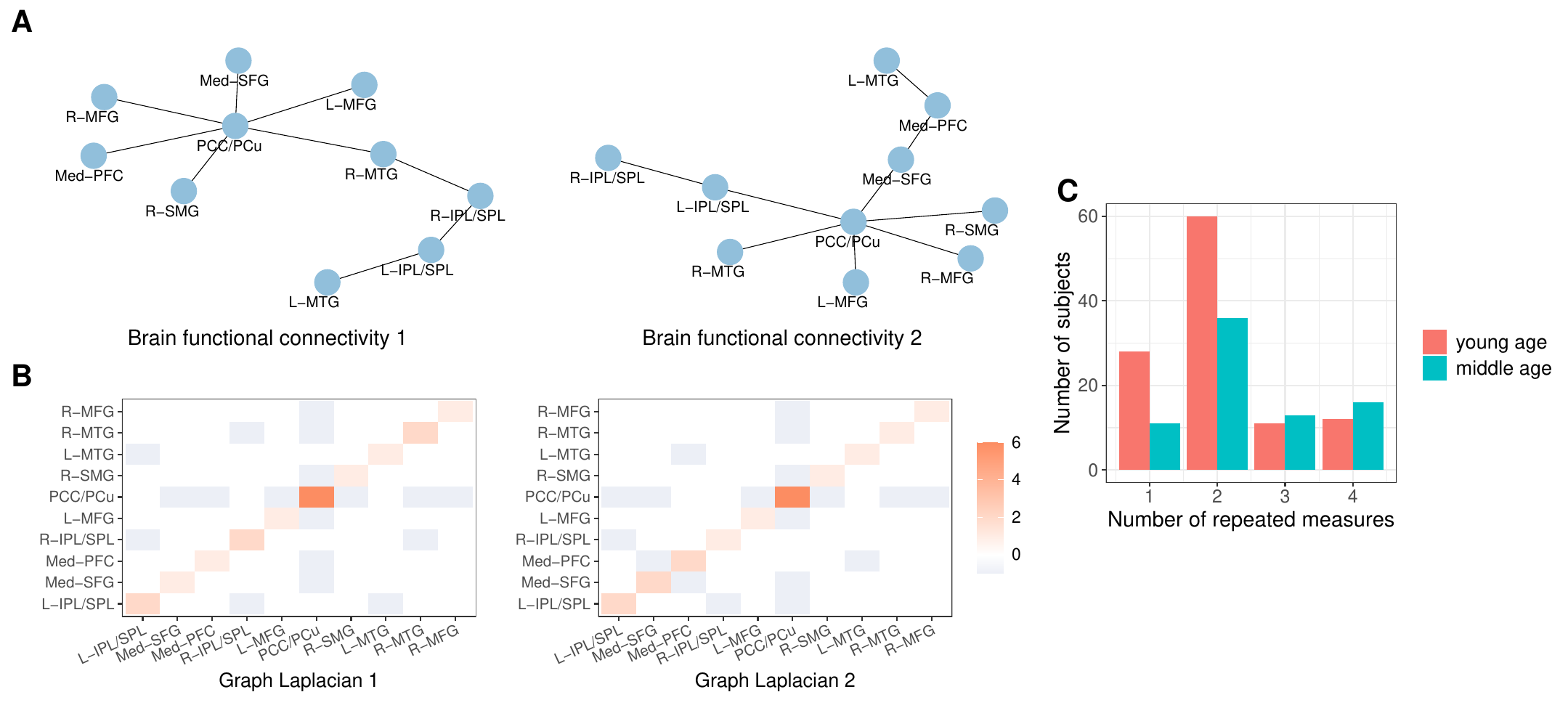}
	\caption{
	(A) Brain functional connectivity networks from two scans of a randomly selected individual, consisting of 10 regions of interest: the left inferior/superior parietal lobule (L-IPL/SPL), medial superior frontal gyrus (Med-SFG), medial prefrontal cortex (Med-PFC), right inferior/superior parietal lobule (R-IPL/SPL), left middle frontal gyrus (L-MFG), posterior cingulate cortex/precuneus (PCC/PCu), right supramarginal gyrus (R-SMG), left middle temporal gyrus (L-MTG), right middle temporal gyrus (R-MTG), and right middle frontal gyrus (R-MFG).
(B) Corresponding graph Laplacians for each network.
(C) Bar plots showing numbers of repeated measures for participants in the young age and middle age groups.
	}\label{fig:example-neuro}
\end{figure}

We focus on addressing the challenges associated with conducting hypothesis testing when the data are repeatedly-observed random objects. The challenges are in several folds: First, non-Euclidean objects lie in general metric spaces without manifold or algebraic structures, rendering many traditional test statistics inapplicable. Consequently, testing methods must rely on metrics that measure pairwise distances between objects. Second, repeated measurements introduce within-subject dependency, and properly utilizing these repeated objects to account for within-subject variability is not straightforward. Moreover, the data are often observed with unequal numbers of repeated measures, adding another layer of complexity to the analysis. As far as we know, few existing methods could effectively handle repeated object data under these scenarios, highlighting a critical gap in the statistical literature.

When analyzing metric space-valued objects, nonparametric tests are preferable because they make minimal assumptions about the data structure. Numerous statistical methods have been proposed that rely on quantifying the similarity or dissimilarity between observations. Examples of such methods include tests based on similarity graphs \citep{friedman1979, schilling1986, henze1988, rosenbaum2005, chen2017new, chen2018weighted}, kernel-based approaches \citep{gretton2012kernel}, and techniques utilizing Fréchet means and variances \citep{dubey2019frechet}. Recently, \cite{zhang2022two} extended graph-based tests to accommodate balanced repeated measurements.

However, these existing tests are not suitable for random objects with unequal repeated measurements. Most of them assume that observations are independent. 
A common approach to addressing this issue is to average the within-subject measurements and use the resulting average as the single observation for each subject \citep{dawson1993size}. 
However, this approach may lead to inflated type 1 errors when the number of repeated measures is imbalanced among subjects (See details in simulation studies).
Moreover, averaging oversimplifies the data's complexity and ignores within-subject variability, which is often of clinical interest in biomedical studies such as physical activity and imaging analysis \citep{murray2020measuring, o2017healthy}. 
\cite{zhang2022two} proposed a test that accommodates repeatedly measured random objects. However, it only applies to settings with balanced repeated measures, making it inapplicable to realistic cases with unequal numbers of repetitions.

In this paper, we aim to develop a new nonparametric test for exchangeable repeated random objects applicable to general metric space-valued data with unequal numbers of repeated measures. 
To account for repeated measurements, we define within-subject variability based on the distances between pairs of repeated measures within each individual. 
By integrating our newly defined concept of within-subject variability with the methodology of \cite{dubey2019frechet}, who used Fréchet means and variances to handle metric space-valued random objects, we propose a generalized Fréchet test to simultaneously detect differences in location, scale, and within-subject variability.
Using empirical process techniques, we establish uniform bounds for the empirical Fréchet variance of the repeated random objects. 
Based on that, we derive the asymptotic distribution of the proposed test statistic, which follows a weighted chi-squared distribution. 
Our approach effectively addresses the limitations of existing methods, providing a robust tool for analyzing complex repeated measures data in metric spaces.

We evaluate the proposed test through comprehensive simulations and applications. Specifically, we apply our method to compare repeatedly measured activity intensity distributions and sleep-related outcomes in the iCOMPARE trial. 
Additionally, we analyze repeatedly measured rs-fMRI brain connectivity data collected by New York University Langone Medical
Center. These evaluations demonstrate the effectiveness of our test in real-world scenarios involving complex, unbalanced repeated measurements.

The remainder of this paper is organized as follows: 
In Section \ref{sec:newfrechet}, we introduce a generalized Fr{\'e}chet test for repeated random objects from multiple populations. 
Section \ref{sec:asymanalysis} establishes the asymptotic distribution of the proposed test. 
In Section \ref{sec:simu}, we assess the performance of our test through simulations. 
The proposed test is demonstrated in Section \ref{sec:realanalysis} with applications to the iCOMPARE trial and rs-fMRI data.
We conclude with a discussion in Section \ref{sec:discuss}.

\section{Nonparametric test for repeated random objects from multiple populations}\label{sec:newfrechet}
\subsection{Fr\'echet analysis for repeated random objects}
We first introduce some necessary notations for defining the metric space of the repeated random objects. Assume that random objects take values in a metric space \((\Omega, d)\), where \(d\) is a metric on \(\Omega\). 
Suppose that for each subject \(i\), there are repeated random objects \(Y_{il}\), \(l = 1, \ldots, r_i\), where \(r_i\) is the number of repeated measures, which may vary among subjects.
Assume that these random objects \(Y_{il}\) (\(i = 1, \ldots, n\), \(l = 1, \ldots, r_i\)) belong to \(k\) different groups \(G_1, G_2, \ldots, G_k\). 
The size of each group is \(n_j\) (\(j = 1, \ldots, k\)), so the total number of subjects is \(n = \sum_{j=1}^k n_j\). 
The number of observations in each group is \(N_j = \sum_{i \in G_j} r_i\) (\(j = 1, \ldots, k\)), and the total number of observations is \(N = \sum_{j=1}^k N_j\).

We further assume that within each group, the repeated random objects have the same marginal distribution, and the joint distributions of any subset of within-subject random objects are identical. 
Since algebraic operations are not generally defined for metric-valued objects, we adopt the Fréchet mean to describe the center of the random objects.
For data from group \(j\), the Fréchet mean \(\mu_j\) is given by
\[
\mu_j=\arg\min_{\omega \in \Omega} \bE\left(d^2(\omega, Y_{il})\right), \quad i \in G_j.
\]
The Fréchet mean generalizes the mean of Euclidean data to metric-valued objects, as it reduces to the expected value of random vectors when \(d\) is the \(L_2\) norm.
The sample Fréchet mean \(\widehat{\mu}_j\) for identically distributed random objects \(Y_{il}\) (\(i \in G_j\), \(l = 1, \ldots, r_i\)) is 
\[
\widehat{\mu}_j = \arg\min_{\omega \in \Omega} \frac{1}{N_j} \sum_{i \in G_j} \sum_{l=1}^{r_i} d^2(\omega, Y_{il}).
\]
The Fréchet variance \(V_j\), defined by
\[
V_j = \bE\left(d^2(\mu_j, Y_{il})\right), \quad i \in G_j,
\]
is a generalization of the variance for Euclidean data. 
It quantifies how far the random objects are spread out from their Fréchet mean. 
The sample Fréchet variance \(\widehat{V}_j\) is
\[
\widehat{V}_j = \frac{1}{N_j} \sum_{i \in G_j} \sum_{l=1}^{r_i} d^2\left(\widehat{\mu}_j, Y_{il}\right).
\]
To measure the variability among the repeated objects within subjects, we introduce the concept of within-subject variability:
\begin{equation*}
\rho_j = \bE\left(d^2(Y_{is}, Y_{it})\right), \quad i \in G_j, \quad s \ne t,
\end{equation*}
which is defined as the expected squared distance between any pair of two within-subject repeated random objects. 
In the special case of Euclidean data with \(d\) being the \(L_2\) norm, \(\rho_j\) is equivalent to twice the common variance minus twice the covariance of two within-subject repeated measures.
The sample within-subject variability \(\widehat{\rho}_j\) is given by
\[
\widehat{\rho}_j = \frac{1}{\sum_{i \in G_j} r_i (r_i - 1)} \sum_{i \in G_j} \sum_{s \ne t} d^2\left(Y_{is}, Y_{it}\right).
\]
We also consider the sample Fréchet mean \(\widehat{\mu}_p\) and sample Fréchet variance \(\widehat{V}_p\) of the pooled data:
\[
\widehat{\mu}_p = \arg\min_{\omega \in \Omega} \frac{1}{N} \sum_{j=1}^k \sum_{i \in G_j} \sum_{l=1}^{r_i} d^2\left(\omega, Y_{il}\right), \quad \widehat{V}_p = \frac{1}{N} \sum_{j=1}^k \sum_{i \in G_j} \sum_{l=1}^{r_i} d^2\left(\widehat{\mu}_p, Y_{il}\right).
\]

\subsection{A Generalized Fr\'echet test for repeated random objects}
We are interested in testing the null hypothesis that the Fr\'echet means, variances, and within-subject variabilities of the population distributions across the \(k\) groups are identical.
Based on the definitions of the Fr\'echet means, variances, and within-subject variabilities, we note that \(V_j\), \(\widehat{V}_j\), \(\rho_j\), and \(\widehat{\rho}_j\) are real-valued, while \(\mu_j\) and \(\widehat{\mu}_j\) are metric-valued objects. 
The Fr\'echet variances and within-subject variabilities among different groups can be compared directly, whereas comparing the Fr\'echet means directly is challenging due to the non-Euclidean nature of the data.
Therefore, we detect differences in the means indirectly by comparing the Fr\'echet variance of the pooled data with a weighted average of the group-specific Fr\'echet variances. 

Let \(\widehat{\lambda}_j = N_j / N\), \(\widehat{u}_{jl} = \widehat{\lambda}_j \widehat{\lambda}_l / (\widehat{\sigma}_j^2 \widehat{\sigma}_l^2)\), and \(\widehat{b}_{jl} = \widehat{\lambda}_j \widehat{\lambda}_l / (\widehat{\gamma}_j^2 \widehat{\gamma}_l^2)\), where \(\widehat{\sigma}_j\) and \(\widehat{\gamma}_j\) are consistent estimators of the asymptotic variances of \(N_j^{1/2} \widehat{V}_j\) and \(N_j^{1/2} \widehat{\rho}_j\), respectively.
We consider the following auxiliary statistics:
\[
D_n = \widehat{V}_p - \sum_{j=1}^k \widehat{\lambda}_j \widehat{V}_j, \quad U_n = \sum_{j<l} \widehat{u}_{jl} (\widehat{V}_j - \widehat{V}_l)^2, \quad R_n = \sum_{j<l} \widehat{b}_{jl} (\widehat{\rho}_j - \widehat{\rho}_l)^2.
\]
Here, \(D_n\) and \(U_n\) are defined similarly to those in \cite{dubey2019frechet}, aiming to detect differences in group means and variances, respectively. 
The statistic \(R_n\) targets differences in within-subject variability among groups.
We then propose the test statistic
\[
Q_n = \frac{N D_n^2}{\sum_{j=1}^k \widehat{\lambda}_j^2 \widehat{\sigma}_j^2} + \frac{N U_n}{\sum_{j=1}^k\widehat{\lambda}_j/\widehat{\sigma}_j^2} + \frac{N R_n}{\sum_{j=1}^k \widehat{\lambda}_j/\widehat{\gamma}_j^2},
\]
which simultaneously detects differences in means, variances, and within-subject variabilities.

The following theorem provides the analytic expressions for \(\widehat{\sigma}_j\) and \(\widehat{\gamma}_j\) so that the proposed test statistic can be computed efficiently. 
Detailed proofs are provided in the Supplementary Material.

\begin{theorem}\label{th:expression}
A consistent estimator of the asymptotic variance of \(N_j^{1/2} \widehat{V}_j\) is given by
\begin{align*}
\widehat{\sigma}_j^2 = \frac{1}{N_j} \sum_{i \in G_j} \left( \sum_{l=1}^{r_i} d^2(\widehat{\mu}_j, Y_{il}) \right)^2 - \frac{\sum_{i \in G_j} r_i^2}{N_j} \widehat{V}_j^2.
\end{align*}
A consistent estimator of the asymptotic variance of \(N_j^{1/2} \widehat{\rho}_j\) is given by
\begin{align*}
\widehat{\gamma}_j^2 = \frac{N_j}{\left( \sum_{i \in G_j} r_i^2 - N_j \right)^2} \sum_{i \in G_j} \left( \sum_{s \ne t} d^2(Y_{is}, Y_{it}) \right)^2 - \frac{N_j \sum_{i \in G_j} r_i^2 (r_i - 1)^2}{\left( \sum_{i \in G_j} r_i^2 - N_j \right)^2} \widehat{\rho}_j^2.
\end{align*}
\end{theorem}

\section{Asymptotic analysis}\label{sec:asymanalysis}
\subsection{Asymptotic null distribution}
In this section, we aim to derive the asymptotic null distribution of the proposed test statistic \(Q_n\). 
Before doing so, we analyze the asymptotic properties of the sample Fr\'echet mean, variance, and within-subject variability, as they are key components of \(Q_n\). 
All proofs are provided in the Supplementary Material.

We begin by stating the conditions.

\begin{condition}\label{assp:unique}
For each group \(j\), the Fr\'echet mean \(\mu_j\) and its estimator \(\widehat{\mu}_j\) exist and are unique. 
Moreover, for any \(\epsilon > 0\),
\[
\inf_{d(\omega, \mu_j) > \epsilon} \bE\left(d^2(\omega, Y_{il}) \right) > \bE\left(d^2(\mu_j, Y_{il}) \right), \quad i \in G_j.
\]
\end{condition}
Condition \ref{assp:unique} is common when establishing the consistency of the sample Fr\'echet mean \(\widehat{\mu}_j\). 
Since \(\widehat{\mu}_j\) is an M-estimator of the empirical process \(M_n(\omega) = \sum_{i \in G_j} \sum_{l=1}^{r_i} d^2(\omega, Y_{il})/N_j\), which converges to the population process \(M(\omega) = \bE(d^2(\omega, Y_{il}))\), Condition \ref{assp:unique} ensures that \(\widehat{\mu}_j\) converges in probability to \(\mu_j\), as implied by Corollary 3.2.3 in \cite{van1996weak}.

Since \(\widehat{\mu}_j\) depends on the observations \(Y_{il}\) (\(i \in G_j\), \(l = 1, \ldots, r_i\)), the quantities \(d^2(\widehat{\mu}_j, Y_{il})\) are not independent. 
This dependence makes it challenging to derive a central limit theorem for the sample Fr\'echet variance. 
Additionally, the dependency arising from repeated measures adds another layer of complexity.
To address the dependence issue, we control the uniform bound
\begin{equation*}
\sup_{\omega \in \mathcal{B}_\delta(\mu_j)} \left| \frac{1}{N_j} \sum_{i \in G_j} \sum_{l=1}^{r_i} d^2(\omega, Y_{il}) - \bE(d^2(\omega, Y_{il})) \right|,
\end{equation*}
where \(\mathcal{B}_\delta(\mu_j) = \{ \omega \in \Omega : d(\omega, \mu_j) < \delta \}\) is a ball of radius \(\delta\) centered at \(\mu_j\) in the metric space \(\Omega\). 
Let \(\mathcal{N}(\mathcal{B}_\delta(\mu_j), d, \epsilon)\) be the covering number of \(\mathcal{B}_\delta(\mu_j)\) using balls of radius \(\epsilon\).
The following lemma shows that the uniform bound can be controlled using the covering number.

\begin{lemma}\label{lem:unibound}
Let \(\mathrm{diam}(\Omega) = \sup_{\omega, \omega' \in \Omega} d(\omega, \omega')\) be the diameter of \(\Omega\). Then
\begin{align*}
&\bE \left(\sup_{\omega \in \mathcal{B}_\delta(\mu_j)} \left| \frac{1}{N_j} \sum_{i \in G_j} \sum_{l=1}^{r_i} d^2(\omega, Y_{il}) - \bE(d^2(\omega, Y_{il})) \right| \right) \\
\leq\; & C \frac{\mathrm{diam}(\Omega) \delta \sqrt{\sum_{i \in G_j} r_i^2}}{N_j} \int_0^1 \sqrt{ \log \mathcal{N}(\mathcal{B}_\delta(\mu_j), d, 2\delta \epsilon) } \, \dd\epsilon,
\end{align*}
where \(C\) is a constant.
\end{lemma}

To further obtain the central limit theorem for \(\widehat{V}_j\), we need the following conditions.
\begin{condition}\label{assp:ri}
There is a uniform lower bound $\widetilde N$ such that $N_j\geq \widetilde N$ for all $j$, and $\widetilde N\rightarrow\infty$ as $n\rightarrow\infty$.
Additionally, there exists a constant \(C\) such that \(r_i \leq C\) for all \(i = 1, \ldots, n\).
\end{condition}

\begin{condition}\label{assp:bound}
The metric space $\Omega$ is bounded; that is, $\mathrm{diam}(\Omega)=\sup_{\omega, \omega' \in \Omega} d(\omega, \omega')<\infty$.
\end{condition}

\begin{condition}\label{assp:complex}
For each group $j$, as \(\delta \to 0\),
\[
\delta \int_0^1 \sqrt{ \log \mathcal{N}(\mathcal{B}_\delta(\mu_j), d, 2\delta \epsilon) } \, \dd\epsilon \to 0.
\]
\end{condition}

\begin{condition}\label{assp:Lyapunov-sigma}
For each group \(j\), let
\[
\sigma_j^2 = \operatorname{Var}\left( d^2(\mu_j, Y_{il}) \right) + (a_r - 1) \operatorname{Cov}\left( d^2(\mu_j, Y_{il}), d^2(\mu_j, Y_{ik}) \right), \quad i \in G_j, \quad l \ne k,
\]
where \(a_r = \lim_{n_j \to \infty} \sum_{i \in G_j} r_i^2 / N_j\).
For some \(\eta > 0\),
\[
\frac{1}{\sigma_j^{2 + \eta} N_j^{1 + \eta / 2}} \sum_{i \in G_j} \bE \left| \sum_{l=1}^{r_i} \left\{d^2(\mu_j, Y_{il}) - \bE(d^2(\mu_j, Y_{il}))\right\} \right|^{2 + \eta} \to 0.
\]
\end{condition}

Condition \ref{assp:ri} is a common assumption that limits the number of repeated measures per subject. 
Condition \ref{assp:bound} requires that the metric space $\Omega$ is bounded.
Condition \ref{assp:complex} constrains the complexity of the metric space \(\Omega\). 
If \(\log \mathcal{N}(\mathcal{B}_\delta(\mu_j), d, \epsilon) \leq C \epsilon^\kappa\) for some constants \(C \geq 0\) and \(\kappa > -2\), then Condition \ref{assp:complex} is satisfied. 
Condition \ref{assp:Lyapunov-sigma} is a Lyapunov-type condition, required to derive the central limit theorem for the sums \(\sum_{l=1}^{r_i} d^2(\mu_j, Y_{il})\), which are independent but not necessarily identically distributed.

By applying empirical process techniques to control the uniform bound in Lemma \ref{lem:unibound} and under the above conditions, we obtain the following central limit theorem for the sample Fr\'echet variance.

\begin{theorem}\label{thm:asym-v}
Under Conditions \ref{assp:unique}--\ref{assp:Lyapunov-sigma},
\[
N_j^{1/2} (\widehat{V}_j - V_j) \xrightarrow{d} N(0, \sigma_j^2).
\]
\end{theorem}

Similar to Condition \ref{assp:Lyapunov-sigma}, the following condition is required to derive the central limit theorem for the sample within-subject variability.

\begin{condition}\label{assp:Lyapunov-gamma}
For each group \(j\), let
\begin{align*}
\gamma_j^2 = & \; a_{\gamma,1} \operatorname{Var}\left( d^2(Y_{is}, Y_{it}) \right) + a_{\gamma,2} \operatorname{Cov}\left( d^2(Y_{is}, Y_{it}), d^2(Y_{is}, Y_{iq}) \right) \\
& + a_{\gamma,3} \operatorname{Cov}\left( d^2(Y_{is}, Y_{it}), d^2(Y_{ip}, Y_{iq}) \right), \quad i \in G_j, \quad s, t, p, q \text{ are distinct},
\end{align*}
where
\[
a_{\gamma,1} = \frac{2}{a_r - 1}, \quad a_{\gamma,2} = \lim_{n_j \to \infty} \frac{4 N_j \sum_{i \in G_j} r_i (r_i - 1)(r_i - 2)}{\left( \sum_{i \in G_j} r_i^2 - N_j \right)^2},
\]
\[
a_{\gamma,3} = \lim_{n_j \to \infty} \frac{N_j \sum_{i \in G_j} r_i (r_i - 1)(r_i - 2)(r_i - 3)}{\left( \sum_{i \in G_j} r_i^2 - N_j \right)^2}.
\]

For some \(\eta > 0\),
\[
\frac{1}{\gamma_j^{2 + \eta} N_j^{1 + \eta / 2}} \sum_{i \in G_j} \bE \left| \sum_{s \ne t} \left\{ d^2(Y_{is}, Y_{it}) - \bE(d^2(Y_{is}, Y_{it}))\right\} \right|^{2 + \eta} \to 0.
\]
\end{condition}

\begin{theorem}\label{th:rho}
Under Conditions \ref{assp:ri} and \ref{assp:Lyapunov-gamma},
\[
N_j^{1/2} (\widehat{\rho}_j - \rho_j) \xrightarrow{d} N(0, \gamma_j^2).
\]
\end{theorem}

The asymptotic distributions of the sample Fr\'echet variance and within-subject variability provided in Theorems \ref{thm:asym-v} and \ref{th:rho} lead to the following asymptotic results for the auxiliary statistics.

\begin{theorem}\label{prop:tests}
Under the null hypothesis of equal Fr\'echet means, variances, and within-subject variabilities across groups, and under Conditions \ref{assp:unique}--\ref{assp:Lyapunov-gamma}, as \(n \to \infty\),
\begin{align*}
N^{1/2} D_n &= o_P(1), \\
\frac{N U_n}{\sum_{j=1}^k \widehat{\lambda}_j / \widehat{\sigma}_j^2} &\xrightarrow{d} \chi^2_{k-1}, \\
\frac{N R_n}{\sum_{j=1}^k \widehat{\lambda}_j / \widehat{\gamma}_j^2} &\xrightarrow{d} \chi^2_{k-1}.
\end{align*}
\end{theorem}

Before deriving the asymptotic distribution of the proposed statistic \(Q_n\), we introduce some auxiliary parameters. 
Let \(\lambda_j = \lim_{n \to \infty} \widehat{\lambda}_j\), \(s_\sigma = (\lambda_1^{1/2} / \sigma_1, \lambda_2^{1/2} / \sigma_2, \ldots, \lambda_k^{1/2} / \sigma_k)^\top\), and define the matrix \(A = I - s_\sigma s_\sigma^\top / (s_\sigma^\top s_\sigma)\), where \(I\) is the \(k \times k\) identity matrix. 
Similarly, let \(s_\gamma = (\lambda_1^{1/2} / \gamma_1, \lambda_2^{1/2} / \gamma_2, \ldots, \lambda_k^{1/2} / \gamma_k)^\top\), and define \(B = I - s_\gamma s_\gamma^\top / (s_\gamma^\top s_\gamma)\).
Let \(\xi_j = \Sigma_{jj} / (\sigma_j \gamma_j)\), where
\begin{align*}
\Sigma_{jj} = & \lim_{n_j \to \infty} \frac{\sum_{i \in G_j} r_i (r_i - 1)(r_i - 2)}{\sum_{i \in G_j} r_i^2 - N_j} \operatorname{Cov}\left( d^2(\mu_j, Y_{il}), d^2(Y_{is}, Y_{it}) \right) \\
& + 2 \operatorname{Cov}\left( d^2(\mu_j, Y_{il}), d^2(Y_{is}, Y_{il}) \right), \quad s, l, t \text{ are distinct}.
\end{align*}

The following theorem states that under the null hypothesis, \(Q_n\) converges in distribution to a weighted sum of chi-squared random variables.

\begin{theorem}\label{th:proptest}
Under the null hypothesis of equal population Fr\'echet means, variances, and within-subject variabilities, and under Conditions \ref{assp:unique}--\ref{assp:Lyapunov-gamma}, as \(n \to \infty\),
\[
Q_n \xrightarrow{d} \sum_{j=1}^{2k - 2} \phi_j \chi^2_1,
\]
where \(\phi_j > 0\) (for \(j = 1, \ldots, 2k - 2\)) are the positive eigenvalues of the matrix
\[
\begin{pmatrix}
A & A \operatorname{diag}(\xi_1, \ldots, \xi_k) B \\
B \operatorname{diag}(\xi_1, \ldots, \xi_k) A & B
\end{pmatrix},
\]
and \(\chi^2_1\) are independent chi-squared random variables with one degree of freedom.
\end{theorem}
Here, the matrices $A$ and $B$ can be estimated by replacing $\lambda_j$, $\sigma_j$, $\gamma_j$ with their sample estimates $\widehat\lambda_j$, $\widehat\sigma_j$, $\widehat\gamma_j$, respectively.
The quantities \(\xi_j\) can be estimated using the finite sample estimator
\[
\widehat{\xi}_j = \frac{\widehat{\Sigma}_{jj} }{ \widehat{\sigma}_j \widehat{\gamma}_j },
\]
where
\begin{align*}
    \widehat{\Sigma}_{jj} = &\frac{1}{ \sum_{i \in G_j} r_i^2 - N_j } \Bigg\{ \sum_{i \in G_j} \left( \sum_{l=1}^{r_i} d^2(\widehat{\mu}_j, Y_{il}) \right) \left( \sum_{s \ne t} d^2(Y_{is}, Y_{it}) \right) \\
    &- \sum_{i \in G_j} r_i^2 (r_i - 1) \widehat{V}_j \widehat{\rho}_j \Bigg\}.
\end{align*}

We reject the null hypothesis of equal Fr\'echet means, variances, and within-subject variabilities when the test statistic \(Q_n\) exceeds the \((1 - \alpha)\)-quantile of the weighted chi-squared distribution \(\sum_{j=1}^{2k - 2} \phi_j \chi^2_1\). 
Denote this quantile by \(q_\alpha\); the rejection region is then
\[
\{ Q_n > q_\alpha \}.
\]

\subsection{Consistency}
\cite{dubey2019frechet} showed that the Fr{\'e}chet test for data without repeated measures is consistent against location and/or scale  alternatives.
Extending their arguments, we demonstrate that the generalized Fréchet test is consistent against any combination of location, scale, and within-subject variability alternatives.

\begin{theorem} \label{th:power}
Under Conditions \ref{assp:unique}-\ref{assp:Lyapunov-gamma}, the proposed test $Q_n$ is consistent against any combination of location, scale, and within-subject variability alternatives in the usual limiting regime.
\end{theorem}
The proof of this theorem is provided in the Supplementary Material.

\subsection{Asymptotics for distributional data}
A common example of random objects in biomedical studies is univariate distributional data, such as distributions of daily activity intensities. 
We consider the 2-Wasserstein distance for such distributional data. 
For two univariate distributions \(F_1\) and \(F_2\) with finite variances, the 2-Wasserstein distance \(d_W\) is given by
$
d_W^2(F_1, F_2) = \int_0^1 \left\{ F_1^{-1}(t) - F_2^{-1}(t) \right\}^2 \dd t,
$
where \(F_1^{-1}\) and \(F_2^{-1}\) are the quantile functions corresponding to \(F_1\) and \(F_2\), respectively.

In practice, the distributions \(Y_{il}\) are rarely observed directly. 
Instead, we have access to random observations \(X_{ilu}\) sampled from these distributions, that is,
\[
X_{ilu} \sim Y_{il}, \quad i = 1, \ldots, n; \quad l = 1, \ldots, r_i; \quad u = 1, \ldots, W_{il},
\]
where \(W_{il}\) is the number of observations for the \(l\)-th measurement of subject \(i\).
To handle the unobserved distributions, we apply nonparametric methods, such as empirical distributions or kernel density estimation \citep{petersen2016functional}, to obtain estimates \(Y_{il}^+\) of the distributions \(Y_{il}\).
Denote by \(\widehat{Q}_n\) the test statistic based on the estimated distributional data. 
Under the following additional regularity conditions, we can derive the asymptotic distribution of \(\widehat{Q}_n\) similar to Theorem \ref{th:proptest}. 
We leave the detailed derivations in the Supplementary Material.

\begin{condition}\label{assp:samplesize}
There is a uniform lower bound $\widetilde W$ such that $W_{il}\geq \widetilde W$ and $\widetilde W\rightarrow\infty$ as $n\rightarrow\infty$.
\end{condition}
\begin{condition}\label{assp:yplus}
There exists a compact interval \(\mathcal{I}\) such that the support of every distribution \(Y_{il}\) is contained within \(\mathcal{I}\). 
Moreover, there is a sequence \(b_N = o(N^{-1/2})\) such that
\[
\sup_{i, l} \sup_{\nu \in \mathcal{W}_2(\mathcal{I})} \bE \left(d_W(Y_{il}^+, Y_{il}) \big| Y_{il} = \nu \right) = O(b_N),
\]
where \(\mathcal{W}_2(\mathcal{I})\) denotes the space of one-dimensional distributions whose supports are contained in \(\mathcal{I}\), equipped with the 2-Wasserstein distance.
\end{condition}

\section{Simulation studies}\label{sec:simu}
We evaluate the performance of the proposed test statistic $Q_n$ under various simulation settings. 
Under each setting, we compare our results with those of the generalized edge-count test proposed by \cite{chen2017new} and the Fréchet test of \cite{dubey2019frechet}. 
Following the recommendations in \cite{chen2017new}, we use the 5-MST structure for the generalized edge-count test.
To the best of our knowledge, neither of these existing methods accommodates data with repeated measures, as both rely on between-subject distance metrics derived from a single observation per subject. 
Applying these tests directly to all repeated observations, without accounting for within-subject correlation, leads to inflated type 1 error rates (results omitted). 
To facilitate a fair comparison, we adapt these tests to the subject level by using the Fréchet mean of each subject's repeated measures as a single representative observation. 
We denote these modified versions of the generalized edge-count test and the Fréchet test as $aS$ and $aF$, respectively.

We consider two groups of subjects, each with an equal size of $n_1=n_2=100$. 
We examine the applicability of $Q_n$ on different types of random objects in a similar way as the simulation settings of \cite{dubey2019frechet}.
For subject $i$, there are $r_i$ repeated random objects.
We generate the $l$-th observation as one of the following: a univariate distribution $(Y_{il}^{(d)})$, a graph Laplacian $(Y_{il}^{(g)})$, a vector $(Y_{il}^{(v)})$, or a combination of these three types $(Y_{il})$. 
Below, we describe the specific data generation settings and present the corresponding results.

In the first setting, we generate $Y_{il}^{(d)}$ as a truncated normal distribution with mean $\theta_{il}$ and variance $\eta_i^2$ over $[-10,10]$. 
We use the 2-Wasserstein distance $d_W$ for these distributional data. 
Let $\eta_i \stackrel{\mathrm{iid}}{\sim} U(1,1.5)$, and let $\theta_{il}\sim N(a_i,1)$ with a subject-specific mean $a_i$.
We assume an exchangeable correlation structure for $\theta_{il}$:
\[
(\theta_{i1},\ldots,\theta_{ir_i})^T|\tilde{\mu}_i \sim N(\tilde{\mu}_i,\tilde{\rho}), \quad \tilde{\mu}_i=a_i1_{r_i}, \quad \tilde{\rho}=\iota 1_{r_i}1_{r_i}^T+(1-\iota)I_{r_i}.
\]
Here, $a_i \stackrel{\mathrm{iid}}{\sim} N(\beta,\epsilon^2)$. 
Thus, $Y_{il}^{(d)}$ depends on $\iota,\beta,\epsilon$.
To examine the type 1 error, we set $\iota=0.5$, $\beta=1$, $\epsilon=1$ for both groups, and set $r_i\equiv2$ for group 1. 
We consider different choices of the number of repeated measures for group 2 with $r_i\equiv r\in\{2,\dots,10\}$. To check the empirical power, we sample $r_i$ from $\{1,2,3\}$ with an equal probability of 1/3 for both groups. 
For group 1, we set $\iota=0.5$, $\beta=1$, $\epsilon=1$. 
For group 2, we vary the values of $\iota$, $\beta$, and $\epsilon$ separately, keeping the other parameters unchanged. We compute the empirical power of the tests for $0\leq\iota\leq 1$, $-1\leq\beta\leq 3$, and $0.5\leq\epsilon\leq 1.5$, corresponding to within-subject variability, Fréchet mean, and Fréchet variance differences, respectively.
Figure \ref{fig:simuden} presents the results.
The proposed test $Q_n$ maintains proper type 1 error control, while $aS$ exhibits slightly inflated type 1 error and $aF$ shows considerably inflated type 1 error when the number of repeated measures differs between groups. 
For within-subject variability differences, $Q_n$ is effective, whereas $aS$ and $aF$ have almost no power. 
All three tests perform similarly for detecting Fr\'echet mean differences. 
For Fr\'echet variance differences, $Q_n$ and $aF$ work well, with $aF$ slightly better, while $aS$ demonstrates the lowest power.

\begin{figure}[!h]
	\centering
	\includegraphics[width=\textwidth]{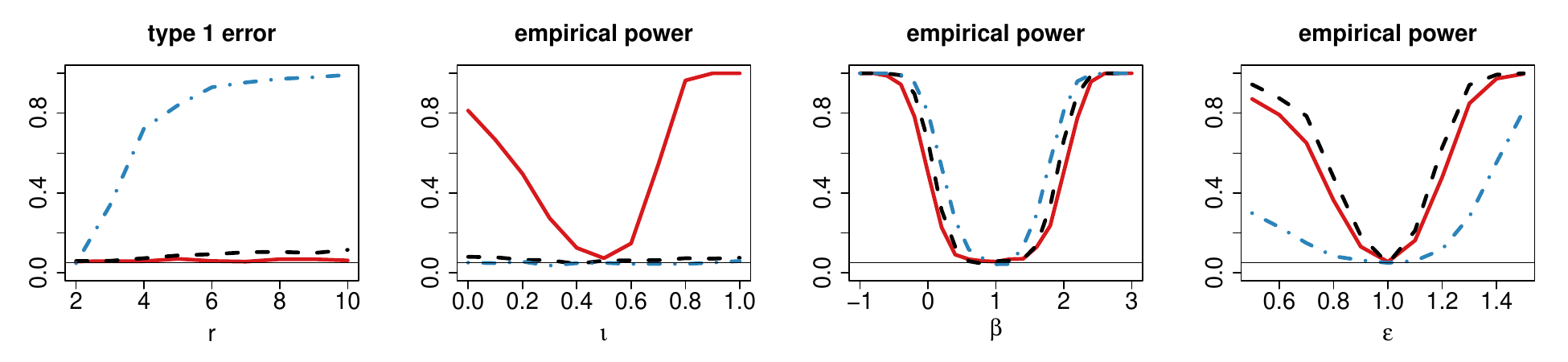}
	\caption{Empirical power for probability distributions $(Y_{il}^{(d)})$ as a function of $r$, $\iota$, $\beta$, and $\epsilon$ (from left to right). The horizontal line indicates the 0.05 significance level. 
	The solid red curve corresponds to the proposed test $Q_n$, and the dashed black and the dot-dashed blue curves represent $aF$ and $aS$, respectively.
	}\label{fig:simuden}
\end{figure}

In the second setting, we consider $Y_{il}^{(g)}$ as a graph Laplacian of a scale free network with 10 nodes. 
The graph Laplacian is defined as $K=E-A$, where $E$ is the degree matrix and $A$ is the adjacency matrix of the graph.
For two graph Laplacians $K_1$ and $K_2$, we use the Frobenius norm $d_F$, with $d_F^2(K_1,K_2)=\mathrm{trace}\{(K_1-K_2)^T(K_1-K_2)\}$.
We construct a scale-free network based on the Barabási-Albert model, where one node is added at each step and the new node forms edges with existing vertices.
For an existing node with $c$ degrees, the probability of connecting to the new node is approximately $P(c)\sim c^{-2\eta_i}$, where $\eta_i\stackrel{\mathrm{iid}}{\sim}U(1,1.5)$.
For each subject $i$, we generate a network $G_i$ using the `sample\_pa()' function from the R package `igraph'. 
We then uniformly sample the $l$-th network for subject $i$ from all possible networks that differ from $G_i$ by $\tau$ edges, making $Y_{il}^{(g)}$ depend on $\tau$.
To examine type 1 error, we set $\tau=3$ for both groups, and $r_i\equiv2$ for group 1. 
For group 2, we let  $r_i\equiv r\in\{2,\dots,10\}$.
To assess the empirical power, we sample $r_i$ from $\{1,2,3\}$ with equal probability for both groups. 
For group 1, we let $\tau=3$.
For group 2, we vary $\tau\in\{1,2,3,4,5\}$ and compute the empirical power to detect within-subject variability differences for network data.
Figure \ref{fig:simugraph} shows that the proposed test $Q_n$ demonstrates very high power with controlled type 1 error when $r_i \neq 2$, while $aF$ and $aS$ exhibit greatly inflated type 1 error and lower power for detecting within-subject variability differences.

\begin{figure}[!h]
	\centering
	\includegraphics[width=0.8\textwidth]{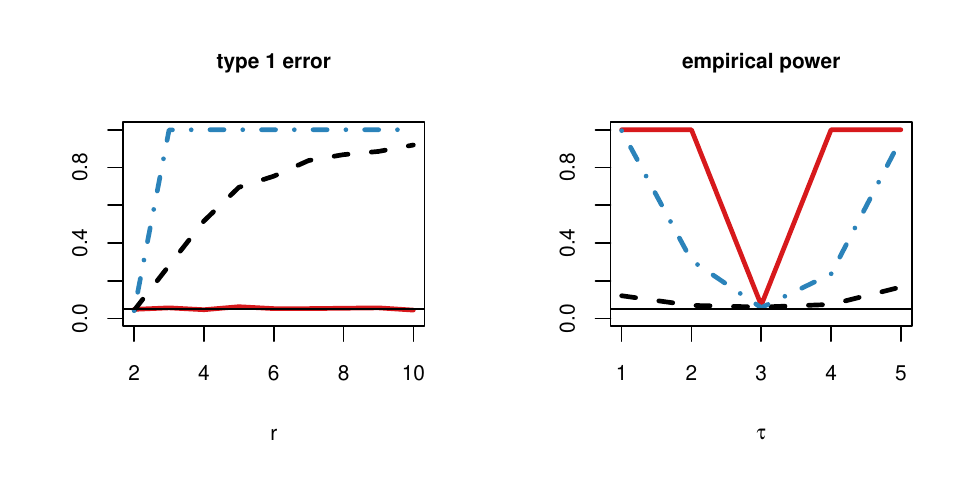}
	\caption{Empirical power for scale-free networks $(Y_{il}^{(g)})$ generated by the Barab{\'a}si-Albert model as a function of $r$ (left) and $\tau$ (right). The horizontal line indicates the 0.05 significance level.
	The solid red curve corresponds to the proposed test $Q_n$, and the dashed black and the dot-dashed blue curves represent $aF$ and $aS$, respectively.}\label{fig:simugraph}
\end{figure}

In the multivariate setting, we consider $Y_{il}^{(v)}$ as a 5-dimensional vector equipped with the $L_2$ norm $d_L^2(X_1,X_2)=(X_1-X_2)^T(X_1-X_2)$.
We assume an exchangeable correlation structure among $Y_{il}^{(v)}$, which gives
\[
(Y_{i1}^{(v)T},\dots,Y_{ir_i}^{(v)T})^T|\tilde\nu_{i}\sim N(\tilde\nu_{i},\tilde\rho\otimes I_5),
\]
where $\otimes$ denotes the Kronecker product, $\tilde\nu_{i}=a_{i}1_{5r_i}$ and $\tilde\rho=\iota 1_{r_i}1_{r_i}^T+(1-\iota)I_{r_i}$.
Here, $a_{i}\stackrel{\mathrm{iid}}{\sim}N(\beta,\epsilon^2)$, and thus $Y_{il}^{(v)}$ depends on $\iota$, $\beta$, $\epsilon$.
We examine the type 1 error and empirical under this setting in a manner similar to the first scenario and observe results consistent with those previously reported. 
Details are provided in the Supplementary Materials.

The last type of random object we study is $Y_{il}$, which consists of the three elements $(Y_{il}^{(d)},Y_{il}^{(g)},Y_{il}^{(v)})$ described above.
For two random objects $Y_s=(Y_s^{(d)},Y_s^{(g)},Y_s^{(v)})$, $s=1,2$, we consider the metric $d^2(Y_1,Y_2)=d_W^2(Y_1^{(d)},Y_2^{(d)})+d_F^2(Y_1^{(g)},Y_2^{(g)})+d_L^2(Y_1^{(v)},Y_2^{(v)})$.
To mimic common factors shared among different elements,
we use the same $\eta_i$ for $Y_{il}^{(d)}$ and $Y_{il}^{(g)}$, and the same $a_{i}$ and $\tilde\rho$ for $Y_{il}^{(d)}$ and $Y_{il}^{(v)}$.
Hence, $Y_{il}$ depends on the values of $\iota$, $\beta$, $\epsilon$, $\tau$.
To examine the type 1 error, we set $\iota=0.5$, $\beta=1$, $\epsilon=1$, $\tau=3$ for both groups, and $r_i\equiv2$ for group 1. 
For group 2, we consider $r_i\equiv r\in\{2,\dots,10\}$.
To check the empirical power, we sample $r_i$ from $\{1,2,3\}$ with equal probability for both groups. 
For group 1, we set $\iota=0.5$, $\beta=1$, $\epsilon=1$, $\tau=3$.
For group 2, we vary $\iota$, $\beta$, $\epsilon$, and $\tau$ seperately, keeping other parameters unchanged.
We compute empirical power for $0\leq\iota\leq 1$, $-1\leq\beta\leq 3$, and $0.5\leq\epsilon\leq 1.5$, and $\tau\in\{1,2,3,4,5\}$, corresponding to within-subject variability $(\iota,\tau)$, Fréchet mean $(\beta)$, and Fréchet variance $(\epsilon)$ differences, respectively.
Figure \ref{fig:simucom} shows that the proposed test $Q_n$ 
performs exceptionally well under all scenarios, maintaining controlled type 1 error. 
However, the tests $aF$ and $aS$ exhibit greatly inflated type 1 error when the number of repeated measures is unequal and show lower power for detecting within-subject variability differences.

\begin{figure}[!h]
	\centering
	\includegraphics[width=\textwidth]{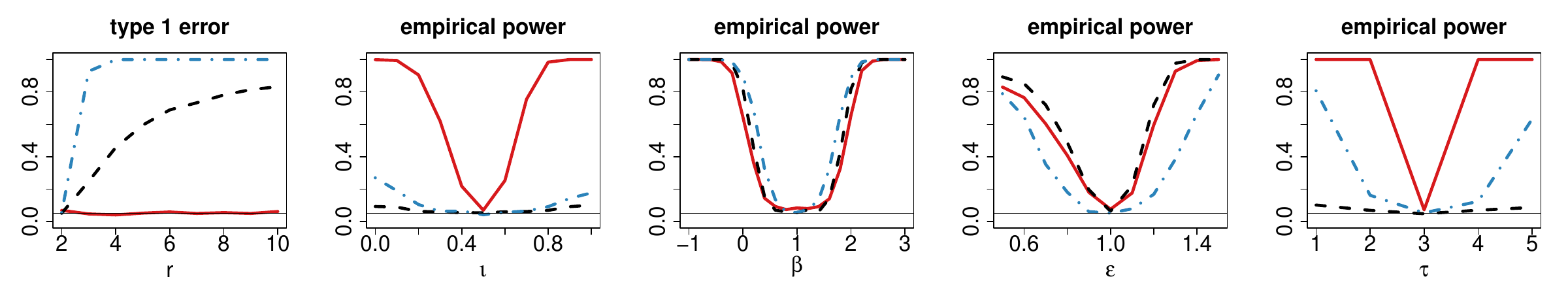}
	\caption{Empirical power for random object $(Y_{il})$ as a function of $r$, $\iota$, $\beta$, $\epsilon$, and $\tau$ (from left to right). The horizontal line indicates the 0.05 significance level. The solid red curve corresponds to the proposed test $Q_n$, and the dashed black and the dot-dashed blue curves represent $aF$ and $aS$, respectively.}\label{fig:simucom}
\end{figure}

\section{Real applications in biomedical studies}\label{sec:realanalysis}
In this section, we apply the proposed  method to two novel data structures: physical activity distributions and multivariate sleep-related outcomes from the iCOMPARE trial, as well as test-retest fMRI connectivity data collected at New York University Langone Medical Center.

To ensure the validity of comparisons, we consider potential limitations of existing methods when handling imbalanced repeated measures. 
The simulation results for $aS$ and $aF$ suggest that using averaged within-subject observations leads to inflated type 1 error rates when the number of repeated measures is imbalanced. 
To address this issue, we resample the data to ensure that each subject has the same number of repeated measures before applying $aS$ and $aF$. Additionally, since the test statistic $S_R$ proposed by \cite{zhang2022two} is designed for data with a balanced repeated measures structure, we also apply $S_R$ to the resampled balanced data. 
The proposed test $Q_n$ is applied directly to the entire dataset.

\subsection{Comparing physical activity distributions and sleep-related outcomes in the iCOMPARE trial}\label{sec:real-icompare}
In the iCOMPARE trial, 395 interns were randomly assigned to either flexible duty-hour programs or standard duty-hour programs \citep{basner2011validity,basner2019sleep}. 
Among them, 203 interns in the flexible duty-hour group worked extended overnight shifts most days and regular day/night shifts several days, while 192 interns in the standard duty-hour group worked regular day/night shifts. 
The study was approved by the Institutional Review Board of the University of Pennsylvania, and all participants provided written informed consent.
During the 14-day iCOMPARE trial, each intern wore an actigraph (model wGT3X-BT, ActiGraph) on the wrist of their non-dominant hand, recording minute-by-minute physical activity. 
Each morning, interns completed a brief smartphone survey that included a sleep log, a sleep quality score, the Karolinska Sleepiness Scale (KSS), and a brief psychomotor vigilance test (PVT). 
Total sleep minutes were derived from actigraphy and sleep logs. 
Thus, each subject had two sets of daily outcomes: minute-by-minute activity counts and a vector of sleep-related outcomes.
After excluding missing records, we retained data from 192 interns in the flexible duty-hour group (1674 daily observations) and 184 interns in the standard duty-hour group (1924 daily observations). 
The number of repeated measures (days) per intern ranged from 1 to 13, resulting in unequal numbers of repeated measures across subjects (Figure \ref{fig:example-icompare}C).

We aim to determine whether duty-hour policies induce significant differences between the two groups in terms of physical activity distributions (P1), sleep-related outcomes (P2), or both (P3).
For subject $i$ on day $l$, 
we characterize physical activity using the distribution of log-transformed activity counts, equipped with the 2-Wasserstein distance $d_W$. 
The sleep-related outcomes are represented by a 4-dimensional vector including sleep duration (hours), sleep quality, alertness (PVT response speed), and sleepiness (KSS), equipped with the $L_2$ norm.

To ensure a fair comparison of $aF$, $aS$, and $S_R$ while maintaining proper type 1 error control, we resample the data so that each intern has the same number of repeated measures. 
Specifically, we consider $r_i \equiv r \in \{10,11,12\}$. 
Interns with fewer than $r$ days are excluded, and for those with more than $r$ days, we randomly select $r$ days. 
This procedure reduces the number of interns in each group, as shown in Table \ref{tab:numicompare}.
To assess the variability in test results due to the choice of days, we repeat the resampling process 500 times. 
Let $p_j$ ($j=1,\ldots,500$) denote the $p$-values obtained from these 500 trials. 
We then combine these values into an overall $p$-value $\widehat{p}$ using the following formula:
\[
\widehat{p} = 1 - \frac{2}{1+\exp(2\theta)},
\]
where
\[
\theta = \frac{1}{500}\sum_{j=1}^{500}\frac{1}{2}\log\left(\frac{1+p_j}{1-p_j}\right).
\]

\begin{table}[!h]
\caption{Number of interns in each group when only considering interns with more than the given number of days}\label{tab:numicompare}
\centering
    \begin{tabular}{l|cc} 
        \hline
        $r$  & flexible duty-hour & standard duty-hour \\
        \hline
        10 & 136 & 90 \\
        11 & 124 & 68 \\
        12 & 100 & 44 \\ 
        \hline
    \end{tabular} 
\end{table}

Figure \ref{fig:icompare} presents boxplots of the $p$-values for $aF$, $aS$, and $S_R$, along with their corresponding overall $p$-values, under different choices for the balanced number of repeated measures. 
Results vary across values of $r$, and even for a fixed $r$, the tests can yield inconsistent conclusions: some $p$-values fall below 0.05 while others exceed 0.05.
Focusing on physical activity (P1) in Figure \ref{fig:icompare}(a), we find that $aS$ and $S_R$ consistently reject the null hypothesis at the 0.05 significance level, while $aF$ yields large, nonsignificant $p$-values for all choices of $r$. 
For sleep-related outcomes (P2), Figure \ref{fig:icompare}(b) shows that only $S_R$ under $r=10$ and $r=11$, and $aF$ under $r=12$, produce overall $p$-values less than 0.05, suggesting that duty-hour policies may not induce a robust, consistent difference in sleep. 
Finally, when considering both the activity distributions and sleep-related outcomes together (P3), Figure \ref{fig:icompare}(c) shows that only $S_R$ consistently rejects the null hypothesis at the 0.05 level.

\begin{figure}[!h]
	\centering
	\includegraphics[width=\textwidth]{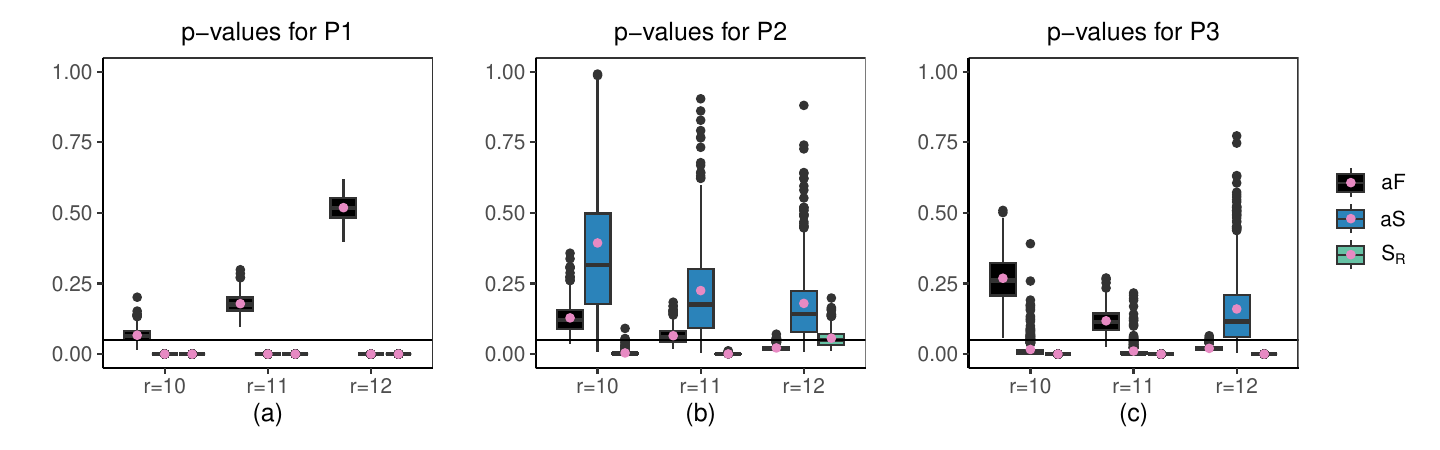}
	\caption{Comparison of physical activity (P1, panel a), sleep-related outcomes (P2, panel b), and both (P3, panel c) between interns in the flexible and standard duty-hour groups. Boxplots display the $p$-values from $aF$, $aS$, and $S_R$ for each comparison. Pink points indicate the overall $p$-values, and the horizontal line marks the 0.05 significance level.}\label{fig:icompare}
\end{figure}

By contrast, the proposed test $Q_n$ applied to the activity distributions (P1) yields an extremely small $p$-value (less than $10^{-3}$), strongly indicating that duty-hour policies significantly affect physical activity patterns. 
When applying $Q_n$ to the sleep-related outcomes (P2), the $p$-value is 0.002, showing that duty-hour policies also significantly influence sleep. 
Considering both the activity distributions and sleep-related outcomes together (P3), the proposed test $Q_n$ again rejects the null hypothesis with a $p$-value less than $10^{-3}$, demonstrating that both activity and sleep are strongly impacted by duty-hour policies.

\subsection{Comparing brain functional connectivity in different age groups}
It has been long reported that age is significantly associated with brain connectivity \citep{dosenbach2010prediction,baghernezhad2024age}. Understanding age-related changes in functional connectivity is essential for predicting brain maturity and identifying potential neurological disorders. 
However, in recent years, the relatively low test-retest reliability of commonly used fMRI-based connectome metrics has raised concerns about the reproducibility of relevant findings, and repeatedly measured neuroimaging data have become increasingly available \citep{anderson2011reproducibility,zuo2014test,zuo2014open,noble2017influences}.

We are interested in assessing whether brain connectivity changes with age with the presence of repeatedly measured neuroimaging data among neurotypical subjects. 
To explore this, we use repeatedly measured resting-state fMRI (rs-fMRI) data from the open-source Consortium for Reliability and Reproducibility (CoRR) (\url{https://fcon_1000.projects.nitrc.org/indi/CoRR/html/index.html}), a collaborative initiative aimed at sharing MRI data and establishing test-retest reliability and reproducibility for commonly used connectome metrics. 
Specifically, we downloaded the dataset from  187 neurotypical individuals (ages 6 to 55) acquired from New York University Langone Medical Center (\url{https://fcon_1000.projects.nitrc.org/indi/CoRR/html/nyu_2.html}). The subjects underwent multiple resting-state scans within several hours, resulting in a total of 415 rs-fMRI scans, with each subject having between 1 and 4  repeated scans (Figure \ref{fig:example-neuro}C).
We preprocess the rs-fMRI data using fMRIPrep \citep{esteban2019fmriprep}, which includes reference image estimation, head-motion correction, slice timing correction, susceptibility distortion correction, and registration to the MNI standard coordinate system. 
These steps ensure that the resulting fMRI volumes are aligned, standardized, and ready for subsequent analyses.

For this analysis, we focus on the 10 cortical hubs listed in Table 3 of \cite{buckner2009cortical}, commonly referred to as regions of interest (ROIs). These hubs include the left inferior/superior parietal lobule (L-IPL/SPL), medial superior frontal gyrus (Med-SFG), medial prefrontal cortex (Med-PFC), right inferior/superior parietal lobule (R-IPL/SPL), left middle frontal gyrus (L-MFG), posterior cingulate cortex/precuneus (PCC/PCu), right supramarginal gyrus (R-SMG), left middle temporal gyrus (L-MTG), right middle temporal gyrus (R-MTG), and right middle frontal gyrus (R-MFG). For each rs-fMRI image, we extract average fMRI signals from $3\times 3\times 3$ cubes centered on the seed voxels of these 10 hubs, yielding a $10\times 10$ correlation matrix based on the pairwise correlations among the ROIs.
Using these correlation matrices, we construct adjacency matrices for the 10 ROIs via minimum spanning trees, a method employed in various studies to investigate neurological conditions and brain network dynamics \citep{guo2017alzheimer,cui2018classification,blomsma2022minimum}. We then derive graph Laplacians from these adjacency matrices to represent the brain’s functional connectivity structure, and measure distances between them using the Frobenius norm.

Subjects are divided into two age groups: a young age group (age $\leq$ 20) and a middle age group (age $>$ 20). 
We have 111 subjects (229 graph Laplacians) in the young age group and 76 subjects (186 graph Laplacians) in the middle age group. 
Applying the proposed test $Q_n$ to the entire dataset yields a $p$-value of 0.039, suggesting moderate evidence that brain functional connectivity differs between the two age groups.

To apply $aF$, $aS$, and $S_R$ under balanced conditions, we set the number of repeated measures to $r_i\equiv 2$. 
Subjects with fewer than 2 scans are excluded, and for those with more than 2 scans, we randomly select 2. 
This subsetting process is repeated 500 times to assess the variability in test results. 
An overall $p$-value is estimated as described in the previous subsection.

\begin{figure}[!h]
	\centering
	\includegraphics[width=\textwidth]{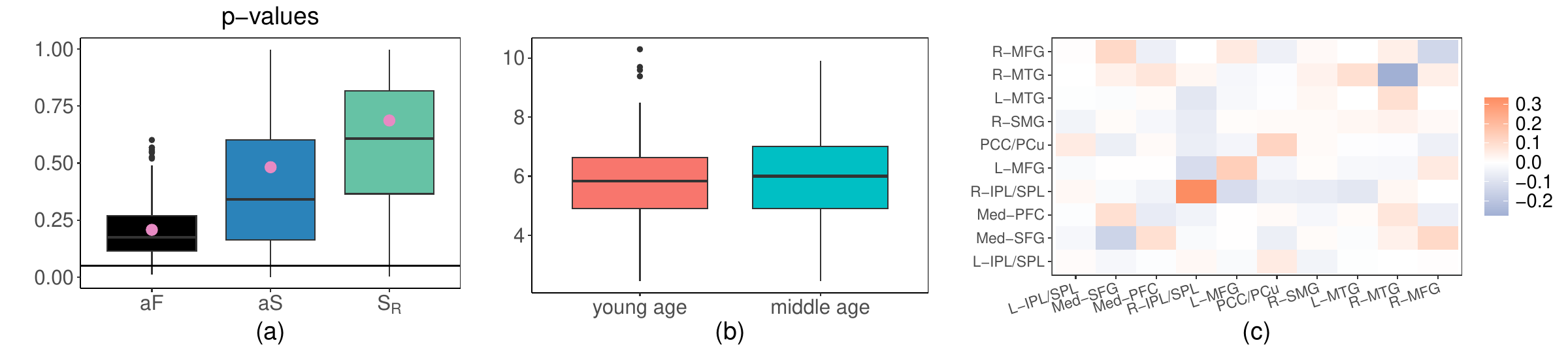}
	\caption{
	Comparison of brain functional connectivity between the young and middle age groups.
	(a) Boxplots of $p$-values from $aF$, $aS$, and $S_R$. 
	Pink points indicate the overall $p$-values, and the horizontal line marks the 0.05 significance level.
	(b) Boxplots of distances between within-subject repeated measures for the young and middle age groups.
	(c) The difference between the average graph Laplacians of the middle and young age groups.
	}\label{fig:neuro}
\end{figure}

Figure \ref{fig:neuro}(a) reveals that although $aF$, $aS$, and $S_R$ produce small $p$-values for certain subsamples, most of these $p$-values and their overall aggregated $p$-values remain large, failing to reject the null hypothesis. 
In contrast, the original data analysis using $Q_n$ indicate a significant difference.
To validate the differences detected by the original analysis, we examine the boxplots of distances between within-subject repeated measures for the young and middle age groups in Figure \ref{fig:neuro}(b). 
We also visualize the difference between the average graph Laplacians of the two age groups in Figure \ref{fig:neuro}(c). 
The exploratory plot reveals an increased average degree in task-related regions such as the R-IPL/SPL in middle-aged individuals compared to younger ones.
This increased connectivity may be attributed to the R-IPL/SPL's roles in attentional control, compensatory mechanisms, and the increased engagement of task-positive networks often observed as a compensatory response during middle age \citep{menardi2024brain}. 
Conversely, the plot shows a decreased average degree in regions such as the R-MTG for middle-aged individuals, which may reflect age-related declines in semantic memory, language processing, and integration within the default mode network \citep{rajah2011age,xu2022structural}.
Together, these figures clearly illustrate distinct patterns in brain functional connectivity between young and middle-aged individuals, demonstrating that the proposed test $Q_n$ can detect meaningful differences in brain networks that other methods may miss.

\section{Conclusion and discussion}\label{sec:discuss}
We consider the testing problem for general metric space-valued data with unequal numbers of repeated measures, which are increasingly encountered in biomedical studies. 
Although some existing nonparametric tests can deal with non-Euclidean objects, testing methods for repeated measures data, particularly with unequal numbers of repeats, remain limited. 
Our numerical studies indicate that simply averaging within-subject measures inflates type 1 error rates, while constraining analyses to subjects with the same number of repeated measures often yields inconsistent results and potentially controversial conclusions.

To address these challenges, we extend the Fréchet test to repeated random objects by explicitly accounting for within-subject dependency. 
We introduce the concept of within-subject variability and propose a generalized Fréchet test that can detect differences in location, scale, and within-subject variability simultaneously. 
Real-world applications, including data from the iCOMPARE trial and neuroimaging studies, demonstrate that our test offers a robust tool for analyzing complex biomedical data, thereby enabling more accurate group comparisons in studies involving repeated measurements.
The within-subject dependency complicates theoretical extensions of existing frameworks, making a direct application of the theory from \cite{dubey2019frechet} inapplicable. 
Here, we employ empirical process techniques to derive the asymptotic weighted chi-squared distribution of our proposed test statistic, which is a nontrivial adaptation beyond the methods in \cite{dubey2019frechet}.

In this paper, we focus on the commonly encountered exchangeable repeated measures structure. 
However, situations with more complex correlation patterns may also occur in practice. 
In such settings, the expected squared distance between within-subject repeated measures could depend on their temporal spacing, necessitating alternative formulations of within-subject variability. 
For instance, under a decaying within-subject correlation structure, one might consider
\[
\bE(d^2(Y_{is}, Y_{it})) = 1/\rho_j^{|s-t|}, \quad i \in G_j, \quad s \neq t,
\]
where $\rho_j\in(0,1]$, to account for  these dependencies. 
Exploring such scenarios would be an interesting topic for future research.

\section*{APPENDIX}

\appendix
\renewcommand{\theequation}{\thesection.\arabic{equation}}

\setcounter{equation}{0}

\section{Proofs of theorems and lemmas}
\subsection{Proof of Theorem \ref{th:expression}}
\begin{proof}
In the following, we only prove the consistency of $\widehat\sigma_j^2$. 
The consistency of $\widehat\gamma_j^2$ can be proved similarly and is omitted.

Note that $\widehat\sigma_j^2$ can be reorganized in the following form
\begin{align*}
\widehat\sigma_j^2 = & \frac{1}{N_j}\sum_{i\in G_j}\sum_{l=1}^{r_i}d^4(\widehat\mu_j,Y_{il})-\frac{\sum_{i \in G_j} r_i^2}{N_j}\left(\frac{1}{N_j}\sum_{i\in G_j}\sum_{l=1}^{r_i}d^2(\widehat\mu_j,Y_{il})\right)^2 \\
& + \frac{1}{N_j}\sum_{i\in G_j}\sum_{l\neq k}d^2(\widehat\mu_j,Y_{il})d^2(\widehat\mu_j,Y_{ik}).
\end{align*}
We will prove the three parts on the right hand are consistent estimators of their corresponding expectations.

(i)
\begin{align*}
& \left|\frac{1}{N_j}\sum_{i\in G_j}\sum_{l=1}^{r_i}d^4(\widehat\mu_j,Y_{il})-\bE (d^4(\mu_j,Y_{il}))\right| \\
= & \left|\frac{1}{N_j}\sum_{i\in G_j}\sum_{l=1}^{r_i}\left(d^4(\widehat\mu_j,Y_{il})-d^4(\mu_j,Y_{il})+d^4(\mu_j,Y_{il})\right)-\bE(d^4(\mu_j,Y_{il}))\right| \\
\leq & 4\mathrm{diam}^3(\Omega)d(\widehat\mu_j,\mu_j)+\left|\frac{1}{N_j}\sum_{i\in G_j}\sum_{l=1}^{r_i}d^4(\mu_j,Y_{il})-\bE(d^4(\mu_j,Y_{il}))\right| \\
=& o_P(1).
\end{align*}

(ii)
\begin{align*}
& \left|\frac{1}{N_j}\sum_{i\in G_j}\sum_{l=1}^{r_i}d^2(\widehat\mu_j,Y_{il})-\bE(d^2(\mu_j,Y_{il}))\right| \\
= & \left|\frac{1}{N_j}\sum_{i\in G_j}\sum_{l=1}^{r_i}\left(d^2(\widehat\mu_j,Y_{il})-d^2(\mu_j,Y_{il})+d^2(\mu_j,Y_{il})\right)-\bE(d^2(\mu_j,Y_{il}))\right| \\
\leq & 2\mathrm{diam}(\Omega)d(\widehat\mu_j,\mu_j)+\left|\frac{1}{N_j}\sum_{i\in G_j}\sum_{l=1}^{r_i}d^2(\mu_j,Y_{il})-\bE(d^2(\mu_j,Y_{il}))\right| \\
=& o_P(1).
\end{align*}

(iii)
\begin{equation}\label{eq:diffcov}
    \begin{aligned}
&\left|\frac{1}{\sum_{i\in G_j}r_i(r_i-1)}\sum_{i\in G_j}\sum_{l\neq k}d^2(\widehat\mu_j,Y_{il})d^2(\widehat\mu_j,Y_{ik})-\bE(d^2(\mu_j,Y_{il})d^2(\mu_j,Y_{ik}))\right|\\
&= |A_1+A_2|,
\end{aligned}
\end{equation}
where 
\[
A_1 = \frac{1}{\sum_{i\in G_j}r_i(r_i-1)}\sum_{i\in G_j}\sum_{l\neq k}\left(d^2(\widehat\mu_j,Y_{il})d^2(\widehat\mu_j,Y_{ik})-d^2(\mu_j,Y_{il})d^2(\mu_j,Y_{ik})\right), 
\]
\[
A_2 = \frac{1}{\sum_{i\in G_j}r_i(r_i-1)}\sum_{i\in G_j}\sum_{l\neq k}d^2(\mu_j,Y_{il})d^2(\mu_j,Y_{ik})-\bE(d^2(\mu_j,Y_{il})d^2(\mu_j,Y_{ik})).
\]
Since 
\begin{align*}
|A_1|\leq & \frac{2\mathrm{diam}^2(\Omega)}{\sum_{i\in G_j}r_i(r_i-1)}\sum_{i\in G_j}\sum_{l\neq k}\left|d(\widehat\mu_j,Y_{il})d(\widehat\mu_j,Y_{ik})-d(\mu_j,Y_{il})d(\mu_j,Y_{ik})\right|\\
= & \frac{2\mathrm{diam}^2(\Omega)}{\sum_{i\in G_j}r_i(r_i-1)}\sum_{i\in G_j}\sum_{l\neq k}\big|d(\widehat\mu_j,Y_{il})d(\widehat\mu_j,Y_{ik})-d(\widehat\mu_j,Y_{il})d(\mu_j,Y_{ik}) \\
& +d(\widehat\mu_j,Y_{il})d(\mu_j,Y_{ik})-d(\mu_j,Y_{il})d(\mu_j,Y_{ik})\big|\\
\leq & 4\mathrm{diam}^3(\Omega)d(\widehat\mu_j,\mu_j)\\
= & o_P(1)
\end{align*}
and $|A_2|=o_P(1)$, the term in \eqref{eq:diffcov} is $o_P(1)$.
Thus, we have $\widehat\sigma_j^2-\sigma_j^2=o_P(1)$.
\end{proof}

\subsection{Proof of Lemma \ref{lem:unibound}}
\begin{proof}
Define 
\[
f_\omega(Y_{ij}) = d^2(\omega,Y_{ij}),
\]
\[
F_\omega(Y_{ij}) = \frac{1}{N_j}\sum_{i\in G_j}\sum_{j=1}^{r_i}f_\omega(Y_{ij}) - \bE(f_\omega(Y_{ij})).
\]
We control the behavior of $F_\omega(Y_{ij})$ uniformly for small $d(\omega,\mu_j)$.
Let 
\[
\cF_j=\{f_\omega(Y):\omega\in\cB_\delta(\mu_j)\},~\cB_\delta(\mu_j)=\{\omega\in \Omega:d(\omega,\mu_j)<\delta\}.
\]
Using Lemma \ref{lemma: symmetric-empirical-process} and Theorem \ref{thm: Dudley}, we have 
\[
\bE\left(\sup_{f_\omega\in\cF_j}|F_\omega(Y_{ij})|\right)\leq C\int_0^D\frac{\sqrt{\log \cN(\cF_j,L^2(\PP_n),t)(\sum_{i\in G_j}r_i^2)}}{N_j}\dd t,
\]
where 
\begin{align*}
D = &\sup_{f,f'\in\cF_j}\|f-f'\|_{L^2} \\
= &\sup_{f,f'\in\cF_j}\sqrt{\frac{1}{N_j}\sum_{i\in G_j}\sum_{j=1}^{r_i}(d^2(\omega,Y_{ij})-d^2(\omega',Y_{ij}))^2} \leq 4\mathrm{diam}(\Omega)\delta.
\end{align*}
Since for $f\in\cF_j$, 
\[
\|f_{\omega_1}-f_{\omega_2}\|_{L^2} \leq 2\mathrm{diam}(\Omega)d(\omega_1,\omega_2),
\]
we have
\[
\cN(\cF_j,L^2(\PP_n),t) \leq \cN(\cB_\delta(\mu),d,\frac{t}{2\mathrm{diam}(\Omega)}).
\]
Hence,
\begin{align*}
E\left(\sup_{f_\omega\in\cF_j}|F_\omega(Y_{ij})|\right)\leq & C\int_0^{4\mathrm{diam}(\Omega)\delta}\frac{1}{N_j}\sqrt{\log \cN(\cB_\delta(\mu),d,\frac{t}{2\mathrm{diam}(\Omega)})(\sum_{i\in G_j}r_i^2)}\dd t\\
=& C\frac{\mathrm{diam}(\Omega)\delta\sqrt{\sum_{i\in G_j}r_i^2}}{N_j}\int_0^1\sqrt{\log\cN(\cB_\delta(\mu),d,2\delta\epsilon)}\dd\epsilon,
\end{align*}
where the value of $C$ may change from line to line.
\end{proof}

\subsection{Proof of Theorem \ref{thm:asym-v}}
\begin{proof}
First, using Lemma \ref{lem:unibound} and under Conditions \ref{assp:unique}-\ref{assp:complex}, we will show that
\begin{equation}\label{prop:diff}
\frac{1}{N_j}\sum_{i\in G_j}\sum_{l=1}^{r_i}\left\{d^2(\widehat\mu_j,Y_{il})-d^2(\mu_j,Y_{il})\right\} = o_P(N_j^{-1/2}).
\end{equation}
That is, we want to prove that for any $\epsilon>0$, $\gamma>0$, there exists $M$ such that for all $N_j\geq M$,
\begin{equation}\label{eq:prop1}
P\left[\left|N_j^{-1/2}\sum_{i\in G_j}\sum_{l=1}^{r_i}\{d^2(\widehat\mu_j,Y_{il})-d^2(\mu_j,Y_{il})\}\right|>\epsilon\right]<\gamma.
\end{equation}
For any small $\delta>0$,
\begin{equation}\label{eq:diff}
\begin{aligned}
 &P\left[\left|N_j^{-1/2}\sum_{i\in G_j}\sum_{l=1}^{r_i}\{d^2(\widehat\mu_j,Y_{il})-d^2(\mu_j,Y_{il})\}\right|>\epsilon\right] \\
&\leq P\left[\left|\frac{1}{N_j}\sum_{i\in G_j}\sum_{l=1}^{r_i}\{d^2(\widehat\mu_j,Y_{il})-d^2(\mu_j,Y_{il})\}\right|>\epsilon N_j^{-1/2}, d(\widehat\mu_j,\mu_j)<\delta\right] \\ 
&\quad\qquad\qquad\qquad+P\{d(\widehat\mu_j,\mu_j)\geq\delta\} . 
\end{aligned} 
\end{equation}
To deal with the first part of \eqref{eq:diff}, we note that when $d(\widehat\mu_j,\mu_j)<\delta$,
\begin{align*}
&\bE(d^2(\widehat\mu_j,Y_{ij})) \\
= & \bE(d^2(\widehat\mu_j,Y_{ij})) - \frac{1}{N_j}\sum_{i\in G_j}\sum_{l=1}^{r_i}d^2(\widehat\mu_j,Y_{ij}) + \frac{1}{N}\sum_{i\in G_j}\sum_{l=1}^{r_i}d^2(\widehat\mu_j,Y_{ij}) \\
\leq & \sup_{f_\omega\in\cF_j}|F_\omega(Y_{ij})| + \frac{1}{N}\sum_{i\in G_j}\sum_{l=1}^{r_i}d^2(\mu_j,Y_{ij}) - \bE(d^2(\mu_j,Y_{ij})) + \bE(d^2(\mu_j,Y_{ij}))\\
\leq & 2\sup_{f_\omega\in\cF_j}|F_\omega(Y_{ij})|+ \bE(d^2(\mu_j,Y_{ij})). 
\end{align*}
We have 
\[
0\leq \bE(d^2(\widehat\mu_j,Y_{ij})) - \bE(d^2(\mu_j,Y_{ij}))\leq  2\sup_{f_\omega\in\cF}|F_\omega(Y_{ij})|,
\]
and 
\begin{align*}
& \left|\frac{1}{N_j}\sum_{i\in G_j}\sum_{l=1}^{r_i}d^2(\widehat\mu_j,Y_{il})-d^2(\mu_j,Y_{il})\right| \\
= & \Bigg|\frac{1}{N_j}\sum_{i\in G_j}\sum_{l=1}^{r_i}d^2(\widehat\mu_j,Y_{il})-\bE(d^2(\widehat\mu_j,Y_{il}))+\bE(d^2(\widehat\mu_j,Y_{il}))-\bE(d^2(\mu_j,Y_{il}))\\
&+\bE(d^2(\mu_j,Y_{il}))- \frac{1}{N_j}\sum_{i\in G_j}\sum_{l=1}^{r_i}d^2(\mu_j,Y_{il})\Bigg|\\
\leq & 4\sup_{f_\omega\in\cF_j}|F_\omega(Y_{ij})|
\end{align*}
for $d(\widehat\mu_j,\mu_j)<\delta$.
Hence,
\begin{align*}
& P\left[\left|\frac{1}{N_j}\sum_{i\in G_j}\sum_{l=1}^{r_i}\{d^2(\widehat\mu_j,Y_{il})-d^2(\mu_j,Y_{il})\}\right|>\epsilon N_j^{-1/2}, d(\widehat\mu_j,\mu_j)<\delta\right]  \\
\leq & P\left[4\sup_{f_\omega\in\cF_j}|F_\omega(Y_{ij})|>\epsilon N_j^{-1/2}\right]  \\
\leq & C\bE\left(\sup_{f_\omega\in\cF_j}|F_\omega(Y_{ij})|\right)\frac{N_j^{1/2}}{\epsilon}  \\
\leq & C\sqrt{\frac{\sum_{i\in G_j}r_i^2}{N_j}}\frac{\mathrm{diam}(\Omega)\delta}{\epsilon}J(\delta)\\
= & \frac{C\delta J(\delta)}{\epsilon}
\end{align*}
where $J(\delta)=\int_0^1\sqrt{\log\cN(\cB_\delta(\mu_j),d,2\delta\epsilon)}\dd\epsilon$ and the value of $C$ may change from line to line.
For any small $\delta>0$ such that $c\delta J(\delta)/\epsilon<\gamma/2$, the expression above can be further bounded by $\gamma/2$.
For any such $\delta$, using the consistency of Fr{\'e}chet mean $\widehat\mu_j$ it is possible to choose $M$ such that the second part of \eqref{eq:diff} can be bounded by $\gamma/2$ for all $N_j\geq M$. Therefore, \eqref{eq:prop1} holds.

Then we have
\[
N_j^{1/2}(\widehat V_j-V_j) = A_1+A_2,
\]
where 
\[
A_1 = N_j^{1/2}\frac{1}{N_j}\sum_{i\in G_j}\sum_{l=1}^{r_i}\{d^2(\widehat\mu_j,Y_{il})-d^2(\mu_j,Y_{il})\}=o_P(1) 
\]
by \eqref{prop:diff}, and 
\[
A_2 = N_j^{1/2}\left[\frac{1}{N_j}\sum_{i\in G_j}\sum_{l=1}^{r_i}d^2(\mu_j,Y_{il})-E\{d^2(\mu_j,Y_{il})\}\right].
\]
Since $\bE(A_2)=0$ and
\begin{align*}
\var(A_2) = & \frac{1}{N_j}\var\left(\sum_{i\in G_j}\sum_{l=1}^{r_i}d^2(\mu_j,Y_{il})\right) \\
= & \frac{1}{N_j}\sum_{i\in G_j}\var\left(\sum_{l=1}^{r_i}d^2(\mu_j,Y_{il})\right) \\
= & \frac{1}{N_j}\sum_{i\in G_j}r_i\var(d^2(\mu_j,Y_{il}))+r_i(r_i-1)\cov(d^2(\mu_j,Y_{il}),d^2(\mu_j,Y_{ik}))\\
= & \var(d^2(\mu_j,Y_{il})) + \left(\sum_{i\in G_j}r_i^2/N_j-1\right)\cov(d^2(\mu_j,Y_{il}),d^2(\mu_j,Y_{ik})),
\end{align*}
the term $A_2$ converges in distribution
to $N(0,\sigma^2)$ by applying the Lyapunov Central Limit Theorem under Condition \ref{assp:Lyapunov-sigma} to the independent random variables $\sum_{l=1}^{r_i}d^2(\mu_j,Y_{il})$, $i\in G_j$. Theorem \ref{thm:asym-v} then follows directly from Slutsky's Theorem. 
\end{proof}

\subsection{Proof of Theorem \ref{th:rho}}
\begin{proof}
Since $E(\widehat\rho_j) = \rho_j$ and
\begin{align*}
& var(\widehat\rho_j) \\
= & \frac{1}{(\sum_{i\in G_j}r_i^2-N_j)^2}\sum_{i\in G_j}var(\sum_{s\neq t}d^2(Y_{is},Y_{it})) \\
= & \frac{1}{(\sum_{i\in G_j}r_i^2-N_j)^2}\sum_{i\in G_j}\Bigg\{\sum_{s\neq t}var(d^2(Y_{is},Y_{it})) + \sum_{\text{different }s,t,q}cov(d^2(Y_{is},Y_{it}),d^2(Y_{is},Y_{iq})) \\
& + \sum_{\text{different }s,t,p,q}cov(d^2(Y_{is},Y_{it}),d^2(Y_{ip},Y_{iq})) \Bigg\} \\
= & \frac{1}{(\sum_{i\in G_j}r_i^2-N_j)^2}\sum_{i\in G_j}\Bigg\{4r_i(r_i-1)(r_i-2)cov(d^2(Y_{is},Y_{it}),d^2(Y_{is},Y_{iq})) \\
& +r_i(r_i-1)(r_i-2)(r_i-3)cov(d^2(Y_{is},Y_{it}),d^2(Y_{ip},Y_{iq})) \\
& +2r_i(r_i-1)var(d^2(Y_{is},Y_{it})) \Bigg\} \\
= & \frac{\sum_{i\in G_j}r_i(r_i-1)(r_i-2)(r_i-3)}{(\sum_{i\in G_j}r_i^2-N_j)^2}cov(d^2(Y_{is},Y_{it}),d^2(Y_{ip},Y_{iq}))\\
& + \frac{4\sum_{i\in G_j}r_i(r_i-1)(r_i-2)}{(\sum_{i\in G_j}r_i^2-N_j)^2}cov(d^2(Y_{is},Y_{it}),d^2(Y_{is},Y_{iq})) \\
& + \frac{2}{\sum_{i\in G_j}r_i^2-N_j}var(d^2(Y_{is},Y_{it})),
\end{align*}
$N_j^{1/2}(\widehat\rho_j-\rho_j)$ converges in distribution
to $N(0,\gamma_j^2)$ by applying the Lyapunov Central Limit Theorem under Assumption \ref{assp:Lyapunov-gamma} to the independent random variables $\sum_{s\neq t}d^2(Y_{is},Y_{it})$, $i\in G_j$.
\end{proof}

\subsection{Proof of Proposition \ref{prop:tests}}
\begin{proof}
(1) Under the null hypothesis of equal population Fr{\'e}chet mean,
\[
\mu_1=\mu_2=\dots=\mu_k=\mu.
\]
We have
\begin{align*}
&N^{1/2}D_n \\
= & N^{-1/2}\sum_{i=1}^n\sum_{l=1}^{r_i}\left(d^2(\widehat\mu_p,Y_{il})-d^2(\mu,Y_{il})\right)- N^{-1/2}\sum_{j=1}^k\sum_{i\in G_j}\sum_{l=1}^{r_i}\left(d^2(\widehat\mu_j,Y_{il})-d^2(\mu,Y_{il})\right)\\
=&o_P(1)
\end{align*}
by \eqref{prop:diff}.

(2) Under the null hypothesis of equal population Fr{\'e}chet variance,
\[
V_1=V_2=\dots=V_k=V.
\]
\begin{align*}
NU_n = & N\sum_{j<l}\frac{\widehat\lambda_j\widehat\lambda_l}{\widehat\sigma_j^2\widehat\sigma_l^2}(\widehat V_j-\widehat V_l)^2 \\
= & N\left\{\sum_{j=1}^k\frac{\widehat\lambda_j}{\widehat\sigma_j^2}\widehat V_j^2\left(\sum_{l=1}^k\frac{\widehat\lambda_l}{\widehat\sigma_l^2}\right)-\left(\sum_{j=1}^k\frac{\widehat\lambda_j}{\widehat\sigma_j^2}\widehat V_j\right)^2\right\} \\
= & N\left(\sum_{l=1}^k\frac{\widehat\lambda_l}{\widehat\sigma_l^2}\right)\widehat V^T\left(\widehat\Lambda_\sigma-\frac{\widehat\lambda_\sigma\widehat\lambda_\sigma^T}{\sum_{j=1}^k\frac{\widehat\lambda_j}{\widehat\sigma_j^2}}\right)\widehat V,
\end{align*}
where $\widehat V=(\widehat V_1,\widehat V_2,\dots,\widehat V_k)^T$, $\widehat\lambda_\sigma=(\widehat\lambda_1/\widehat\sigma_1^2,\widehat\lambda_2/\widehat\sigma_2^2,\dots,$  $\widehat\lambda_k/\widehat\sigma_k^2)^T$, $\widehat\Lambda_\sigma=diag(\widehat\lambda_\sigma)$.
Let $Z_n = N^{1/2}\widehat\Lambda_\sigma^{1/2}\widehat V$ and $\widehat s_\sigma=(\widehat\lambda_1^{1/2}/\widehat\sigma_1,\widehat\lambda_2^{1/2}/\widehat\sigma_2,\dots,\widehat\lambda_k^{1/2}/\widehat\sigma_k)^T$.
Therefore,
\begin{align}\label{eq:nu}
\frac{NU_n}{\sum_{j=1}^k\widehat\lambda_j/\widehat\sigma_j^2} = N\widehat V^T\left(\widehat\Lambda_\sigma-\frac{\widehat\lambda_\sigma\widehat\lambda_\sigma^T}{\sum_{j=1}^k\frac{\widehat\lambda_j}{\widehat\sigma_j^2}}\right)\widehat V = Z_n^TZ_n-\frac{Z_n^T\widehat s_\sigma\widehat s_\sigma^TZ_n}{\widehat s_\sigma^T\widehat s_\sigma}. 
\end{align}

Since $\widehat\sigma_j\rightarrow\sigma_j$ in probability and $\widehat\lambda_j\rightarrow\lambda_j$, $\widehat\lambda_j/\widehat\sigma_j^2\rightarrow\lambda_j/\sigma_j^2$.
Let $\lambda_\sigma=(\lambda_1/\sigma_1^2,\lambda_2/\sigma_2^2,\dots,$ $\lambda_k/\sigma_k^2)^T$, $s_\sigma=(\lambda_1^{1/2}/\sigma_1,\lambda_2^{1/2}/\sigma_2,\dots,\lambda_k^{1/2}/\sigma_k)^T$, $\Lambda_\sigma=diag(\lambda_\sigma)$, and $Z=N^{1/2}\Lambda_\sigma^{1/2}\widehat V$.
Thus, $\widehat\lambda_\sigma\rightarrow\lambda_\sigma$, $\widehat s_\sigma\rightarrow s_\sigma$, and $\widehat\Lambda_\sigma\rightarrow\Lambda_\sigma$.
Continuing from \eqref{eq:nu}, we see that the limiting distribution of $NU_n/\sum_{j=1}^k(\widehat\lambda_j/\widehat\sigma_j^2)$ is the same as that of
\begin{equation}\label{eq:nu2}
 Z^T(I-s_\sigma s_\sigma^T/s_\sigma^Ts_\sigma)Z=(AZ)^TAZ,
\end{equation}
where $A=I-s_\sigma s_\sigma^T/s_\sigma^Ts_\sigma$.
Since $AZ\rightarrow N(0,A)$ in distribution and $A$ is a projection matrix with $k-1$ ones eigenvalues and 1 zero eigenvalue, the limiting distribution of \eqref{eq:nu2} is $\chi^2_{k-1}$.

(3) We prove (3) by using the similar techniques in (2).
Let $\widehat \rho=(\widehat \rho_1,\widehat \rho_2,\dots,\widehat \rho_k)^T$, $\widehat\lambda_\gamma=(\widehat\lambda_1/\widehat\gamma_1^2,$ $\widehat\lambda_2/\widehat\gamma_2^2,\dots,\widehat\lambda_k/\widehat\gamma_k^2)^T$, $\widehat s_\gamma=(\widehat\lambda_1^{1/2}/\widehat\gamma_1,\widehat\lambda_2^{1/2}/\widehat\gamma_2,\dots,\widehat\lambda_k^{1/2}/\widehat\gamma_k)^T$, $\widehat\Lambda_\gamma=diag(\widehat\lambda_\gamma)$, and $W_n = N^{1/2}\widehat\Lambda_\gamma^{1/2}\widehat \rho$.
Let $\lambda_\gamma=(\lambda_1/\gamma_1^2,\lambda_2/\gamma_2^2,\dots,\lambda_k/\gamma_k^2)^T$, $s_\gamma=(\lambda_1^{1/2}/\gamma_1,\lambda_2^{1/2}/\gamma_2,\dots,$ $\lambda_k^{1/2}/\gamma_k)^T$, $\Lambda_\gamma=diag(\lambda_\gamma)$, and $W=N^{1/2}\Lambda_\gamma^{1/2}\widehat \rho$.
Thus, $\widehat\lambda_\gamma\rightarrow\lambda_\gamma$, $\widehat s_\gamma\rightarrow s_\gamma$, and $\widehat\Lambda_\gamma\rightarrow\Lambda_\gamma$.
The limiting distribution of $NR_n/\sum_{j=1}^k(\widehat\lambda_j/\widehat\gamma_j^2)$ is the same as that of
\begin{equation}\label{eq:nr}
 W^T(I-s_\gamma s_\gamma^T/s_\gamma^Ts_\gamma)W=(BW)^TBW,
\end{equation}
where $B=I-s_\gamma s_\gamma^T/s_\gamma^Ts_\gamma$.
Since $BW\rightarrow N(0,B)$ in distribution and $B$ is a projection matrix with $k-1$ ones eigenvalues and 1 zero eigenvalue, the limiting distribution of \eqref{eq:nr} is $\chi^2_{k-1}$.
\end{proof}

\subsection{Proof of Theorem \ref{th:proptest}}
\begin{proof}
First, Proposition \ref{prop:tests} implies that 
\[
\frac{ND_n^2}{\sum_{j=1}^k\widehat\lambda_j^2\widehat\sigma_j^2} = o_P(1).
\]

Since the limiting distribution of 
\begin{equation}\label{eq:nu+nr}
\frac{NU_n}{\sum_{j=1}^k\widehat\lambda_j/\widehat\sigma_j^2} + \frac{NR_n}{\sum_{j=1}^k\widehat\lambda_j/\widehat\gamma_j^2}
\end{equation}
is the same as that of 
\[
(AZ)^TAZ+(BW)^TBW,
\]
we analyze the joint distribution of $((AZ)^T,(BW)^T)^T$ in \eqref{eq:nu2} and \eqref{eq:nr}.
Let 
\[
\widetilde V_j = \frac{1}{N_j}\sum_{i\in G_j}\sum_{l=1}^{r_i}d^2(\mu_j,Y_{il}),~\widetilde V = (\widetilde V_1,\widetilde V_2,\dots,\widetilde V_k)^T,\text{ and }\widetilde Z = N^{1/2}\Lambda_\sigma^{1/2}\widetilde V.
\]
We have
\begin{align*}
&cov(\widetilde V_j,\widehat\rho_j) \\
= & cov\left(\frac{1}{N_j}\sum_{i\in G_j}\sum_{l=1}^{r_i}d^2(\mu_j,Y_{il}),\frac{1}{\sum_{i\in G_j}r_i^2-N_j}\sum_{i \in G_j}\sum_{s\neq t}d^2(Y_{is},Y_{it})\right) \\
= & \frac{1}{N_j(\sum_{i\in G_j}r_i^2-N_j)}\sum_{i\in G_j} cov\left(\sum_{l=1}^{r_i}d^2(\mu_j,Y_{il}),\sum_{s\neq t}d^2(Y_{is},Y_{it})\right) \\
 = & \frac{1}{N_j(\sum_{i\in G_j}r_i^2-N_j)}\sum_{i\in G_j}\big\{2r_i(r_i-1)cov\left(d^2(\mu_j,Y_{il}),d^2(Y_{is},Y_{il})\right) \\
 &+ r_i(r_i-1)(r_i-2)cov\left(d^2(\mu_j,Y_{il}),d^2(Y_{is},Y_{it})\right)\big\} \\
 = & \frac{\sum_{i\in G_j}r_i(r_i-1)(r_i-2)}{N_j(\sum_{i\in G_j}r_i^2-N_j)}cov\left(d^2(\mu_j,Y_{il}),d^2(Y_{is},Y_{it})\right) \\
 & + \frac{2}{N_j}cov\left(d^2(\mu_j,Y_{il}),d^2(Y_{is},Y_{il})\right).
\end{align*}
Let \(\xi_j = \Sigma_{jj} / (\sigma_j \gamma_j)\), where
\begin{align*}
\Sigma_{jj} = & \lim_{n_j \to \infty} \frac{\sum_{i \in G_j} r_i (r_i - 1)(r_i - 2)}{\sum_{i \in G_j} r_i^2 - N_j} \operatorname{Cov}\left( d^2(\mu_j, Y_{il}), d^2(Y_{is}, Y_{it}) \right) \\
& + 2 \operatorname{Cov}\left( d^2(\mu_j, Y_{il}), d^2(Y_{is}, Y_{il}) \right), \quad s, l, t \text{ are distinct}.
\end{align*}
Since
\[
cov(\widetilde V_j,\widehat\rho_l) = 0 \text{ when }j\neq l,
\]
we have
\begin{align}\label{eq:zw}
cov(\widetilde Z,W) = & N\Lambda_\sigma^{1/2}cov(\widetilde V,\widehat\rho)\Lambda_\gamma^{1/2} = diag(\xi_1,\dots,\xi_k)=O(1).
\end{align}
Note that 
\[
cov(Z-\widetilde Z,W) = N^{1/2}\Lambda_\sigma^{1/2}cov(\widehat V-\widetilde V,W),
\]
where $\widehat V-\widetilde V=o_P(N^{-1/2})$ by \eqref{prop:diff} and $W=O_P(1)$. We have 
\begin{equation}\label{eq:zzw}
cov(Z-\widetilde Z,W)=o(1).
\end{equation}
Hence, \eqref{eq:zw} and \eqref{eq:zzw} imply
\begin{align*}
cov(Z,W) = & cov(\widetilde Z+Z-\widetilde Z,W)\\
= & diag(\xi_1,\dots,\xi_k)+o(1).
\end{align*}
Since
\[
E((AZ)^T,(BW)^T)^T)=0,
\]
\[
cov((AZ)^T,(BW)^T)^T)=\begin{pmatrix}
A & Adiag(\xi_1,\dots,\xi_k)B\\
Bdiag(\xi_1,\dots,\xi_k)A & B
\end{pmatrix}+o(1),
\]
\[
((AZ)^T,(BW)^T)^T\rightarrow N\left(0,\begin{pmatrix}
A & Adiag(\xi_1,\dots,\xi_k)B\\
Bdiag(\xi_1,\dots,\xi_k)A & B
\end{pmatrix}\right) 
\]
in distribution.
Continuing from \eqref{eq:nu+nr}, we see that
\[
\frac{NU_n}{\sum_{j=1}^k\widehat\lambda_j/\widehat\sigma_j^2} + \frac{NR_n}{\sum_{j=1}^k\widehat\lambda_j/\widehat\gamma_j^2}\rightarrow\sum_{j=1}^{2k-2}\phi_j\chi_1^2,
\] 
where $\phi_j$ are positive eigenvalues of $\begin{pmatrix}
A & Adiag(\xi_1,\dots,\xi_k)B\\
Bdiag(\xi_1,\dots,\xi_k)A & B
\end{pmatrix}$.
\end{proof}

\subsection{Proof of Theorem \ref{th:power}}
\begin{proof}
Let 
\[
V_j^0 = \widehat V_j-V_j,\quad \rho_j^0=\widehat\rho_j-\rho_j.
\]
The statistics $U_n$ and $R_n$ can be rewrited as
\[
U_n=\sum_{j<l}\frac{\widehat\lambda_j\widehat\lambda_l}{\widehat\sigma_j^2\widehat\sigma_l^2}(\widehat V_j-\widehat V_l)^2
= U_n^0+\Delta_U, 
\]
\[
R_n=\sum_{j<l}\frac{\widehat\lambda_j\widehat\lambda_l}{\widehat\gamma_j^2\widehat\gamma_l^2}(\widehat \rho_j-\widehat \rho_l)^2
= R_n^0+\Delta_R, 
\]
where 
\[
U_n^0 = \sum_{j<l}\frac{\widehat\lambda_j\widehat\lambda_l}{\widehat\sigma_j^2\widehat\sigma_l^2}(V^0_j-V^0_l)^2,
\]
\[
\Delta_U = \sum_{j<l}\frac{\widehat\lambda_j\widehat\lambda_l}{\widehat\sigma_j^2\widehat\sigma_l^2}(V_j-V_l)^2+2\sum_{j<l}\frac{\widehat\lambda_j\widehat\lambda_l}{\widehat\sigma_j^2\widehat\sigma_l^2}(V^0_j-V^0_l)(V_j-V_l),
\]
\[
R_n^0 = \sum_{j<l}\frac{\widehat\lambda_j\widehat\lambda_l}{\widehat\gamma_j^2\widehat\gamma_l^2}(\rho^0_j-\rho^0_l)^2,
\]
\[
\Delta_R = \sum_{j<l}\frac{\widehat\lambda_j\widehat\lambda_l}{\widehat\gamma_j^2\widehat\gamma_l^2}(\rho_j-\rho_l)^2+2\sum_{j<l}\frac{\widehat\lambda_j\widehat\lambda_l}{\widehat\gamma_j^2\widehat\gamma_l^2}(\rho^0_j-\rho^0_l)(\rho_j-\rho_l).
\]
Hence,
\begin{align*}
Q_n = & \frac{N D_n^2}{\sum_{j=1}^k \widehat{\lambda}_j^2 \widehat{\sigma}_j^2} + \frac{N U_n}{\sum_{j=1}^k\widehat{\lambda}_j/\widehat{\sigma}_j^2} + \frac{N R_n}{\sum_{j=1}^k \widehat{\lambda}_j/\widehat{\gamma}_j^2} \\
= & \frac{N D_n^2}{\sum_{j=1}^k \widehat{\lambda}_j^2 \widehat{\sigma}_j^2} + \frac{N (U_n^0+\Delta_U)}{\sum_{j=1}^k \widehat{\lambda}_j/\widehat{\sigma}_j^2} + \frac{N(R_n^0+\Delta_R)}{\sum_{j=1}^k\widehat{\lambda}_j/\widehat{\gamma}_j^2}
\end{align*}

Note that 
\[
\frac{D_n^2}{\sum_{j=1}^k\widehat\lambda_j^2\widehat\sigma_j^2}+\frac{\Delta_U}{\sum_{j=1}^k\widehat\lambda_j/\widehat\sigma_j^2}+\frac{\Delta_R}{\sum_{j=1}^k\widehat\lambda_j/\widehat\gamma_j^2} 
\] 
is a uniformly consistent estimator of 
\[
\frac{D^2}{\sum_{j=1}^k\lambda_j^2\sigma_j^2}+\frac{U}{\sum_{j=1}^k\lambda_j/\sigma_j^2}+\frac{R}{\sum_{j=1}^k\lambda_j/\gamma_j^2}, 
\]
where 
\[
D = V_p-\sum_{j=1}^k\lambda_jV_j,~U=\sum_{j<l}\frac{\lambda_j\lambda_l}{\sigma_j^2\sigma_l^2}(V_j-V_l)^2,~R=\sum_{j<l}\frac{\lambda_j\lambda_l}{\gamma_j^2\gamma_l^2}(\rho_j-\rho_l)^2.
\]
We have 
\begin{align*}
\lim_{n\rightarrow\infty}\frac{Q_n}{N}= &\frac{D^2}{\sum_{j=1}^k\lambda_j^2\sigma_j^2}+\frac{U}{\sum_{j=1}^k\lambda_j/\sigma_j^2}+\frac{R}{\sum_{j=1}^k\lambda_j/\gamma_j^2} \\
&+ \lim_{n\rightarrow\infty}\left(\frac{U_n^0}{\sum_{j=1}^k \widehat{\lambda}_j/\widehat{\sigma}_j^2} + \frac{R_n^0}{\sum_{j=1}^k\widehat{\lambda}_j/\widehat{\gamma}_j^2}\right).
\end{align*}
When any of the location, scale, or within-subject variability alternatives exist, i.e., \(D \neq 0\), \(U > 0\), or \(R > 0\), the limit of \(Q_n / N\) is strictly positive, indicating the consistency of the proposed test \(Q_n\).
\end{proof}

\section{Asymptotic analysis for random univariate distributional data}
We use $Y_{il}^+$ to denote the estimate of true distribution $Y_{il}$ using a random sample $X_{ilu}~(u=1,\dots,W_{il})$.
To construct $\widehat Q_n$ based on $Y_{il}^+$, we define the following quantities in a similar way as those in the main text.
The sample Fr{\'e}chet mean $\widehat\mu_j^+$, sample Fr{\'e}chet variance $\widehat V_j^+$, and sample within-subject variability $\widehat\rho_j^+$ of group $j$ are
\[
\widehat\mu_j^+ = \arg\min_{\omega\in\Omega}\frac{1}{N_j}\sum_{i\in G_j}\sum_{l=1}^{r_i}d_W^2(\omega,Y_{il}^+)=\left(\frac{1}{N_j}\sum_{i\in G_j}\sum_{l=1}^{r_i}(Y_{il}^+)^{-1}\right)^{-1},
\]
\[
\widehat V_j^+=\frac{1}{N_j}\sum_{i\in G_j}\sum_{l=1}^{r_i}d_W^2(\widehat\mu_j^+,Y_{il}^+),
\]
\[
\widehat\rho_j^+ = \frac{1}{\sum_{i\in G_j}r_i^2-N_j}\sum_{i\in G_j}\sum_{s\neq t}d_W^2(Y_{is}^+,Y_{it}^+),\text{ respectively.}
\]
The variance estimate of $N_j^{1/2}\widehat V_j^+$ is given by

\begin{align*}
\widehat\sigma_j^{+2} = & \frac{1}{N_j}\sum_{i\in G_j}\left(\sum_{l=1}^{r_i}d_W^2(\widehat\mu_j^+,Y_{il}^+)\right)^2-\frac{\sum_{i\in G_j}r_i^2}{N_j}\widehat V_j^{+2}. 
\end{align*}
The variance estimate of $N^{1/2}_j\widehat\rho_j^+$ is given by
\begin{align*}
\widehat\gamma_j^{+2} = & \frac{N_j}{(\sum_{i\in G_j}r_i^2-N_j)^2}\sum_{i\in G_j}\left(\sum_{s\neq t}d_W^2(Y_{is}^+,Y_{it}^+)\right)^2 - \frac{N_j\sum_{i\in G_j}r_i^2(r_i-1)^2}{(\sum_{i\in G_j}r_i^2-N_j)^2}\widehat\rho_j^{+2}.
\end{align*}
Let $\widehat\mu_p^+$ be the pooled sample Fr{\'e}chet mean and $\widehat V_p^+$ be the corresponding pooled sample Fr{\'e}chet variance,
\[
\widehat\mu_p^+ = \arg\min_{\omega\in\Omega}\frac{1}{N}\sum_{j=1}^k\sum_{i\in G_j}\sum_{l=1}^{r_i}d_W^2(\omega,Y_{il}^+),~\widehat V_p^+=\frac{1}{N}\sum_{j=1}^k\sum_{i\in G_j}\sum_{l=1}^{r_i}d_W^2(\widehat\mu_p^+,Y_{il}^+).
\]

Let 
\[
D_n^+ = \widehat V_p^+-\sum_{j=1}^k\widehat\lambda_j\widehat V_j^+,~U_n^+=\sum_{j<l}\widehat u_{jl}(\widehat V_j^+-\widehat V_l^+)^2,~R_n^+=\sum_{j<l}\widehat b_{jl}(\widehat \rho_j^+-\widehat \rho_l^+)^2,  
\]
and 
\[
Q_n^+ = \frac{ND_n^{+2}}{\sum_{j=1}^k\widehat\lambda_j^2\widehat\sigma_j^{+2}} + \frac{NU_n^+}{\sum_{j=1}^k\widehat\lambda_j/\widehat\sigma_j^{+2}} + \frac{NR_n^+}{\sum_{j=1}^k\widehat\lambda_j/\widehat\gamma_j^{+2}}.
\]
where $\widehat\lambda_j=N_j/N$, $\widehat u_{jl}=\widehat\lambda_j\widehat\lambda_l/\widehat\sigma_j^{+2}\widehat\sigma_l^{+2}$, and $\widehat b_{jl}=\widehat\lambda_j\widehat\lambda_l/\widehat\gamma_j^{+2}\widehat\gamma_l^{+2}$.

Now $\xi_j$ can be estimated using finite sample estimate
\[
\widehat\xi_j^+ = N\widehat\lambda_j\widehat\Sigma_{jj}^+/(\widehat\sigma_j^+\widehat\gamma_j^+),
\]
where
\begin{align*}
\widehat\Sigma_{jj}^+ = &\frac{1}{N_j(\sum_{i\in G_j}r_i^2-N_j)}\Bigg\{\sum_{i\in G_j}\sum_ld^2(\widehat\mu_j^+,Y_{il}^+)\sum_{s,t}d^2(Y_{is}^+,Y_{it}^+)\\
&-\sum_{i\in G_j}r_i^2(r_i-1)\widehat V_j^+\widehat\rho_j^+\Bigg\}.
\end{align*}
\begin{theorem}\label{th:proptest2}
Under the null hypothesis of equal population Fr\'echet means, variances, and within-subject variabilities, and under Conditions \ref{assp:unique}-\ref{assp:yplus}, as $n\rightarrow\infty$, 
\[
Q_n^+\rightarrow \sum_{j=1}^{2k-2}\phi_j\chi_1^2 \text{ in distribution},
\]
where $\phi_j~(j=1,\dots,2k-2)$ are positive eigenvalues of 
\[\begin{pmatrix}
A & Adiag(\xi_1,\dots,\xi_k)B\\
Bdiag(\xi_1,\dots,\xi_k)A & B
\end{pmatrix}
\] and \(\chi^2_1\) are independent chi-squared random variables with one degree of freedom.
\end{theorem}
\begin{proof}
We will show that under the conditions of Theorem \ref{th:proptest2}, 
\begin{align*}
& (1) ~\frac{1}{N_j}\sum_{i\in G_j}\sum_{l=1}^{r_i}\left\{d_W^2(\widehat\mu_j^+,Y_{il}^+)-d_W^2(\mu_j,Y_{il})\right\} = o_P(N_j^{-1/2}), \\
& (2) ~N_j^{1/2}(\widehat V_j^+-V_j)\rightarrow N(0,\sigma_j^2) \text{ in distribution}, \\
& (3) ~N_j^{1/2}(\widehat\rho_j^+-\rho_j)\rightarrow N(0,\gamma_j^2) \text{ in distribution}.
\end{align*}
Then Theorem \ref{th:proptest2} can be proved using similar techniques as in the proof of Theorem \ref{th:proptest}.

Proof of (1).
Since
\[
\widehat\mu_j^+ = \left(\frac{1}{N_j}\sum_{i\in G_j}\sum_{l=1}^{r_i}(Y_{il}^+)^{-1}\right)^{-1},~\widehat\mu_j = \left(\frac{1}{N_j}\sum_{i\in G_j}\sum_{l=1}^{r_i}(Y_{il})^{-1}\right)^{-1},
\]
\[
d_W(\widehat\mu_j^+,\widehat\mu_j) = \left\|\frac{1}{N_j}\sum_{i\in G_j}\sum_{l=1}^{r_i}\left((Y_{il}^+)^{-1}-(Y_{il})^{-1}\right)\right\|_2\leq\sup_{Y_{il}}d_W(Y_{il}^+,Y_{il})=O_P(b_W).
\]
Then
\begin{align*}
& \frac{1}{N_j}\sum_{i\in G_j}\sum_{l=1}^{r_i}\left\{d_W^2(\widehat\mu_j^+,Y_{il}^+)-d_W^2(\widehat\mu_j,Y_{il})\right\} \\
= & \frac{1}{N_j}\sum_{i\in G_j}\sum_{l=1}^{r_i}\left\{d_W^2(\widehat\mu_j^+,Y_{il}^+)-d_W^2(\widehat\mu_j^+,Y_{il})+d_W^2(\widehat\mu_j^+,Y_{il})-d_W^2(\widehat\mu_j,Y_{il})\right\} \\
\leq & \frac{1}{N_j}\sum_{i\in G_j}\sum_{l=1}^{r_i} 2diam(\Omega)(d_W(Y_{il}^+,Y_{il})+d_W(\widehat\mu_j^+,\widehat\mu_j))=O_P(b_W).
\end{align*}
Hence,
\begin{align*}
& \frac{1}{N_j}\sum_{i\in G_j}\sum_{l=1}^{r_i}\left\{d_W^2(\widehat\mu_j^+,Y_{il}^+)-d_W^2(\mu_j,Y_{il})\right\} \\
= & \frac{1}{N_j}\sum_{i\in G_j}\sum_{l=1}^{r_i}\left\{d_W^2(\widehat\mu_j^+,Y_{il}^+)-d_W^2(\widehat\mu_j,Y_{il})+d_W^2(\widehat\mu_j,Y_{il})-d_W^2(\mu_j,Y_{il})\right\}=o_P(N_j^{-1/2}).
\end{align*}

Proof of (2).
\begin{align*}
& N_j^{1/2}(\widehat V_j^+-V_j) \\
=& N_j^{1/2}\frac{1}{N_j}\sum_{i\in G_j}\sum_{l=1}^{r_i}\left\{d_W^2(\widehat\mu_j^+,Y_{il}^+)-d_W^2(\mu_j,Y_{il})+d_W^2(\mu_j,Y_{il})-E[d_W^2(\mu_j,Y)]\right\} \\
=& \frac{1}{N_j^{1/2}}\sum_{i\in G_j}\sum_{l=1}^{r_i}\{d_W^2(\mu_j,Y_{il})-E[d_W^2(\mu_j,Y)]\}+o_P(1)\\
\rightarrow & N(0,\sigma^2),
\end{align*}
where (1) is applied in the second step.

Proof of (3).
Since 
\begin{align*}
&N_j^{1/2}(\widehat\rho_j-\widehat\rho_j^+) \\
=&N_j^{1/2}\frac{1}{\sum_{i\in G_j}r_i^2-N_j}\sum_{i\in G_j}\sum_{s\neq t}\left(d_W^2(Y_{is},Y_{it})-d_W^2(Y_{is}^+,Y_{it}^+)\right) \\
=&N_j^{1/2}\frac{1}{\sum_{i\in G_j}r_i^2-N_j}\sum_{i\in G_j}\sum_{s\neq t}\big(d_W^2(Y_{is},Y_{it})-d_W^2(Y_{is},Y_{it}^+)+d_W^2(Y_{is},Y_{it}^+)\\
&-d_W^2(Y_{is}^+,Y_{it}^+)\big)\\
\leq& N^{1/2}2diam(\Omega)(d_W(Y_{it},Y_{it}^+)+d_W(Y_{is},Y_{is}^+))=o_P(1),
\end{align*}
Theorem \ref{th:rho} implies $N_j^{1/2}(\widehat\rho_j^+-\rho_j)\rightarrow N(0,\gamma_j^2)$ in distribution.
\end{proof}

\section{Technical lemmas}
\begin{lemma}[Symmetrization of empirical processes]\label{lemma: symmetric-empirical-process}
Suppose that for random variables $X_{ij}$, $i=1,\dots,n,$ $j=1,\dots,r_i$, $X_{ij}$ and $X_{st}$ are independent when $i\neq s$, and possibly nonindependent when $i=s$. Let $N=\sum_{i=1}^nr_i$.
\[
    \EE\sup_{f\in\cF}\left(\frac{1}{N} \sum_{i=1}^n \sum_{j=1}^{r_i}f(X_{ij}) - \EE f(X)\right)\leq 2\EE\sup_{f\in\cF}\left(\frac{1}{N} \sum_{i=1}^n\xi_i\sum_{j=1}^{r_i} f(X_{ij})\right),
\]
where $\xi_1,\dots,\xi_n$ are \iid Rademacher random variable: $\PP(\xi_i=1)=\PP(\xi_i=-1)=\frac 1 2$
\end{lemma}

\begin{proof}
Let $X_{ij}'$ be an independent copy of $X_{ij}$. To simplify the notation, we use $\EE_{X}$ and $\EE_{X'}$ to denote the expectation \wrt $\{X_{ij}\}$ and $\{X_{ij}'\}$, respectively.
Then,
\begin{align}
&\EE\sup_{f\in\cF}\left(\frac{1}{N} \sum_{i=1}^n\sum_{j=1}^{r_i} f(X_{ij}) - \EE f(X)\right) \\
=& \EE_{X}\sup_{f\in\cF}\EE_{X'}\left(\frac{1}{N} \sum_{i=1}^n\sum_{j=1}^{r_i} f(X_{ij}) - f(X_{ij}')\right) \\
\leq& \EE_{X,X'}\sup_{f\in\cF}\left(\frac{1}{N} \sum_{i=1}^n\sum_{j=1}^{r_i}f(X_{ij}) - f(X_{ij}')\right)
\end{align}
Due to that $\sum_{j=1}^{r_i}f(X_{ij})-f(X_{ij}')$ is symmetric, for any $\{\xi_i\}\in \{\pm 1\}^n$, we have 
\begin{align*}
&\EE_{X,X'}\sup_{f\in\cF}\left(\frac{1}{N} \sum_{i=1}^n\sum_{j=1}^{r_i} f(X_{ij}) - f(X_{ij}')\right)\\
=& \EE_{X,X'}\sup_{f\in\cF}\frac{1}{N} \sum_{i=1}^n \xi_i\sum_{j=1}^{r_i}(f(X_{ij}) - f(X_{ij}')) \\
=& \EE_{X,X',\xi}\sup_{f\in\cF}\frac{1}{N} \sum_{i=1}^n \xi_i\sum_{j=1}^{r_i}(f(X_{ij}) - f(X_{ij}'))\\
\leq & \EE_{X,X',\xi}\left(\sup_{f\in\cF}\frac{1}{N} \sum_{i=1}^n \xi_i\sum_{j=1}^{r_i}f(X_{ij}) + \sup_{f\in\cF}\frac{1}{N} \sum_{i=1}^n -\xi_i\sum_{j=1}^{r_i} f(X'_{ij})\right)\\
=& 2 \EE_{X_i,\xi}\sup_{f\in\cF}\frac{1}{N} \sum_{i=1}^n \xi_i\sum_{j=1}^{r_i}f(X_{ij})
\end{align*}
\end{proof}

\begin{definition}[Modified Rademacher complexity]
The empirical Rademacher complexity of a function class $\cF$ on repeated measures data $\{X_{ij}\}$ is defined as 
\[
\erad(\cF) = \EE_{\xi}\left(\sup_{f\in\cF}\frac{1}{N} \sum_{i=1}^n \xi_i\sum_{j=1}^{r_i} f(X_{ij})\right).
\]
The population Rademacher complexity is given by 
\[
    \rad(\cF) = \EE_{X}(\erad(\cF)).
\]
\end{definition}

The symmetrization Lemma \ref{lemma: symmetric-empirical-process} implies that 
\begin{equation}
\EE\sup_{f\in\cF}\left(\frac{1}{N} \sum_{i=1}^n\sum_{j=1}^{r_i} f(X_{ij}) - \EE f(X)\right)\leq 2\rad(\cF).
\end{equation}

\begin{lemma}[Massart's lemma]
Assume that $\sup_{x\in \cX,f\in\cF}|f(x)|\leq B$ and $\cF$ is finite. Then, 
\[
    \erad(\cF)\leq \frac{B}{N}\sqrt{2\log |\cF|\sum_{i=1}^nr_i^2}.
\]
\end{lemma}
\begin{proof}
Let $Z_f = \sumin \xi_i\sum_{j=1}^{r_i} f(X_{ij})$. Then, 
\begin{align*}
\log \EE(e^{\lambda Z_f}) & = \log\left(\prod_{i=1}^n \EE(e^{\lambda \xi_i \sum_{j=1}^{r_i}f(X_{ij})})\right) \\
& = \sum_{i=1}^n \log\EE e^{\lambda \xi_i \sum_{j=1}^{r_i}f(X_{ij})}\stackrel{(i)}{\leq} \sum_{i=1}^n \lambda^2 \frac{(r_iB-(-r_iB))^2}{8} = \frac{\lambda^2B^2}{2}\sum_{i=1}^nr_i^2,
\end{align*}
where $(i)$ follows from the Hoeffding's lemma, which provides an upper bound of the log-moment generating functions of a bounded random variable.
Since 
\[
\sup_{X,f}|\xi_i\sum_{j=1}^{r_i}f(X_{ij})|\leq\sum_{j=1}^{r_i}\sup_{X,f}|f(X_{ij})|\leq r_iB,
\]
$Z_f$ is sub-Gaussian with the variance proxy $\sigma^2=B^2\sum_{i=1}^nr_i^2$. Using the maximal inequality, we have 
\begin{align}
\erad(\cF) = \frac{1}{N}\EE_{\xi}(\sup_{f\in\cF} Z_f)\leq \frac{B}{N}\sqrt{2\log|\cF|\sum_{i=1}^nr_i^2}.
\end{align}
\end{proof}

Consider the function space $(\cF, L^2(\PP_n))$, where $\cF$ is the hypothesis class and $L^2(\PP_n)$ is defined by 
\[
    \|f-f'\|_{L^2(\PP_n)} = \sqrt{\frac{1}{N}\sum_{i=1}^n\sum_{j=1}^{r_i} (f(x_{ij})-f'(x_{ij}))^2},
\]
where $x_{ij}(i=1,\dots,n,j=1,\dots,r_i)$ denote the finite training samples. 

\begin{theorem}[Dudley's integral inequality]\label{thm: Dudley}
Let $D=\sup_{f,f'\in \cF}\|f-f'\|_{L^2(\PP_n)}$ be the diameter of $\cF$. Then,
\[
\erad(\cF)\leq 12 \int_0^D\frac{\sqrt{\log \cN(\cF,L^2(\PP_n),t)(\sum_{i=1}^nr_i^2)}}{N}\dd t.
\]
\end{theorem}
\begin{proof}
Let $\varepsilon_j=2^{-j}D$ be the dyadic scale and $\cF_j$ be an $\varepsilon_j$-cover of $\cF$.
Given $f\in\cF$, let $f_j\in \cF_j$ such that $\|f_j-f\|\leq \varepsilon_j$.
Consider the decomposition 
\begin{equation}\label{eqn: lec2-multi-resoluation}
f = f - f_m + \sum_{j=1}^m (f_j-f_{j-1}),
\end{equation}
where $f_0=0$.
Notice that 
\begin{itemize}
\item $\|f-f_m\|\leq \varepsilon_m$. 
\item $\|f_j-f_{j-1}\|\leq \|f_j-f\|+\|f-f_{j-1}\|\leq \varepsilon_j + \varepsilon_{j-1}\leq 3\varepsilon_{j}$.
\end{itemize}
Let $\xi=(\xi_{11},\dots,\xi_{1r_1},\dots,\xi_{n1},\dots,\xi_{nr_n})^T$, where $\xi_{ij}=\xi_i$.
Then, 
\begin{align*}
\erad(\cF) &= \EE_{\xi}\left(\sup_{f\in\cF}\frac{1}{N} \sum_{i=1}^n \xi_i\sum_{j=1}^{r_i} f(X_{ij})\right) \\
&= \EE_\xi \sup_{f\in \cF}\langle \xi, f\rangle \\
&= \EE_\xi \sup_{f\in \cF}\left(\langle \xi, f-f_m\rangle + \sum_{j=1}^m \langle \xi,f_j-f_{j-1}\rangle\right)\\
&\leq \varepsilon_m + \EE_\xi\sup_{f\in\cF}\sum_{j=1}^m \langle \xi,f_{j}-f_{j-1}\rangle\\
&\leq \varepsilon_m + \sum_{j=1}^m \EE_\xi\sup_{f\in\cF}\langle\xi, f_j - f_{j-1}\rangle\\
& = \varepsilon_m + \sum_{j=1}^m \EE_\xi\sup_{f_j\in \cF_j, f_{j-1}\in \cF_{j-1}}\langle\xi, f_j - f_{j-1}\rangle \\
&=\varepsilon_m + \sum_{j=1}^m \erad(\cF_j\cup\cF_{j-1}).
\end{align*}
Using the Massart lemma and the fact that $\sup_{f\in \cF_j,f'\in \cF_{j-1}}\|f_j-f_{j-1}\|\leq 3\varepsilon_{j}$,
\begin{align*}
\erad(\cF) &\leq \varepsilon_m + \sum_{j=1}^m \frac{3\varepsilon_{j}}{N}\sqrt{2\log(|\cF_j||\cF_{j-1}|)(\sum_{i=1}^nr_i^2)}\\
&\leq \varepsilon_m + \sum_{j=1}^m \frac{6\varepsilon_j}{N} \sqrt{\log|\cF_j|(\sum_{i=1}^nr_i^2)} \\
& = \varepsilon_m + \sum_{j=1}^m \frac{12(\varepsilon_{j}-\varepsilon_{j+1})}{N} \sqrt{\log \cN(\cF,L^2(\PP_n),\varepsilon_j)(\sum_{i=1}^nr_i^2)}.
\end{align*}
Taking $m\to\infty$, we obtain 
\[
\erad(\cF)\leq 12 \int_{0}^D  \frac{\sqrt{\log \cN(\cF,L^2(\PP_n),t)(\sum_{i=1}^nr_i^2)}}{N}\dd t.
\]
\end{proof}

\section{Additional simulation details for the multivariate setting}
To examine the type 1 error, we set $\iota=0.5$, $\beta=1$, and $\epsilon=1$ for both groups, and let $r_i\equiv2$ for group 1. 
For group 2, we consider $r_i\equiv r\in\{2,\dots,10\}$.
To assess the empirical power, we sample $r_i$ from $\{1,2,3\}$ equally for both groups. 
For group 1, we set $\iota=0.5$, $\beta=1$, and $\epsilon=1$.
For group 2, we vary $\iota$, $\beta$, and $\epsilon$ seperately, keeping other parameters fixed.
We compute the empirical power for $0\leq\iota\leq 1$, $-1\leq\beta\leq 3$, and $0.5\leq\epsilon\leq 1.5$, corresponding to within-subject variability, Fr{\'e}chet mean, and Fr{\'e}chet variance differences, respectively.
Figure \ref{fig:simuvec} yields results similar to the first setting, demonstrating the effectiveness of the proposed test $Q_n$.

\begin{figure}[!h]
	\centering
	\includegraphics[width=\textwidth]{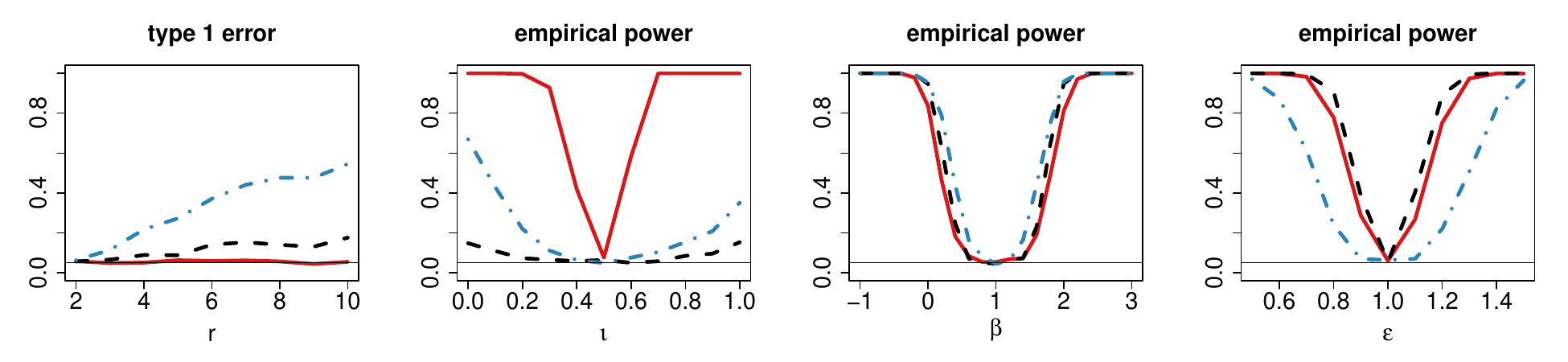}
	\caption{Empirical power for 5-dimensional vectors $(Y_{il}^{(v)})$ as a function of $r_i$, $\iota$, $\beta$, and $\epsilon$ (from left to right). The horizontal line indicates the 0.05 significance level. The solid red curve corresponds to the proposed test $Q_n$, and the dashed black and the dot-dashed blue curves represent $aF$ and $aS$, respectively.}\label{fig:simuvec}
\end{figure}

\bibliographystyle{chicago} 
\bibliography{repfre}       

\end{document}